\documentclass[11pt]{article}
\usepackage{geometry}
\usepackage[T1]{fontenc}
\usepackage{graphicx}
\usepackage[linesnumbered,boxed,ruled,vlined,algo2e,resetcount,algosection]{algorithm2e}
\usepackage{xspace}
\usepackage[utf8]{inputenc}
\usepackage[noadjust]{cite}
\usepackage{amsmath, amssymb, amsfonts, amsthm,mathtools}
\usepackage{enumerate}
\usepackage{xcolor}
\usepackage[pagebackref=true,colorlinks,linkcolor=magenta,citecolor=blue,bookmarks,bookmarksopen,bookmarksnumbered]{hyperref}
\usepackage{bbm}
\usepackage{ebproof}
\usepackage{paralist}
\usepackage{thmtools}
\usepackage{rotating}
\usepackage{mdframed}
\usepackage{multirow}
\usepackage{fancyvrb}
\usepackage{mathrsfs}
\usepackage{indentfirst}
\usepackage{mleftright}
\usepackage[nottoc]{tocbibind}%
\usepackage[langfont=roman, classfont=bold]{complexity}
\usepackage{tablefootnote}
\usepackage{relsize}
\usepackage{exscale}
\usepackage{tikz}
\usetikzlibrary{positioning}
\usepackage{afterpage}

\usepackage{enumitem}

\usepackage[capitalize]{cleveref}%

\usepackage{centernot}

\newtheorem{theorem}{Theorem}[section]
\newtheorem{lemma}[theorem]{Lemma}
\newtheorem{corollary}[theorem]{Corollary}
\newtheorem{claim}[theorem]{Claim}

\newtheorem{fact}[theorem]{Fact}

\newtheorem{conjecture}[theorem]{Conjecture}

\newtheorem{proposition}[theorem]{Proposition}

\declaretheoremstyle[
spaceabove=\topsep, spacebelow=\topsep,
headfont=\normalfont\bfseries,
notefont=\bfseries, notebraces={}{},
bodyfont=\normalfont\itshape,
postheadspace=0.5em,
name={\ignorespaces},
numbered=no,
headpunct=.]
{mystyle}

\theoremstyle{definition}
\newtheorem{definition}[theorem]{Definition}

\newtheorem{assumption}[theorem]{Assumption}

\theoremstyle{remark}

\makeatletter
\def\moverlay{\mathpalette\mov@rlay}
\def\mov@rlay#1#2{\leavevmode\vtop{%
		\baselineskip\z@skip \lineskiplimit-\maxdimen
		\ialign{\hfil$\m@th#1##$\hfil\cr#2\crcr}}}
\newcommand{\charfusion}[3][\mathord]{
	#1{\ifx#1\mathop\vphantom{#2}\fi
		\mathpalette\mov@rlay{#2\cr#3}
	}
	\ifx#1\mathop\expandafter\displaylimits\fi}
\makeatother

\makeatletter
\renewenvironment{algomathdisplay}
 {\[}
 {\@endalgocfline\vspace{-\baselineskip}\]{\DontPrintSemicolon\;}}
\makeatother

\DeclareMathOperator{\GF}{GF}

\renewcommand{\poly}{\mathrm{poly}}

\renewcommand{\polylog}{\mathrm{polylog}}

\newcommand{\eps}{\varepsilon}

\newcommand{\Enc}{\mathsf{Enc}}
\newcommand{\Dec}{\mathsf{Dec}}

\newcommand{\Ext}{\mathsf{Ext}}

\newlang{\MCSP}{MCSP}
\newlang{\MOCSP}{MOCSP}
\newlang{\MFSP}{MFSP}
\newlang{\MKtP}{MKtP}
\newlang{\MKTP}{MKTP}
\newlang{\itrMCSP}{itrMCSP}
\newlang{\itrMKTP}{itrMKTP}
\newlang{\itrMINKT}{itrMINKT}
\newlang{\MINKT}{MINKT}
\newlang{\MINK}{MINK}
\newlang{\MINcKT}{MINcKT}
\newlang{\CMD}{CMD}
\newlang{\DCMD}{DCMD}
\newlang{\CGL}{CGL}
\newlang{\PARITY}{PARITY}
\renewlang{\Gap}{Gap}
\newlang{\Avoid}{\textnormal{\textsc{Avoid}}}
\newlang{\LossyCode}{\textsc{Lossy-Code}}
\newlang{\MissingString}{\textsc{Missing-String}}
\newlang{\SinkOfDAG}{\textsc{Sink-Of-DAG}}
\newlang{\Iter}{\textsc{Iter}}
\newlang{\Palindromes}{\textsc{Palindromes}}
\newlang{\Sparsification}{\textsc{Sparsification}}
\newlang{\HamEst}{\mathsf{HammingEst}}
\newlang{\HamHit}{\mathsf{HammingHit}}
\newlang{\CktEval}{\textsc{Circuit-Eval}}
\newlang{\Hard}{\textsc{Hard}}
\newlang{\cHard}{\textsc{cHard}}
\newlang{\Lin}{\textsc{Lin}}
\newlang{\CAPP}{CAPP}
\newlang{\GapUNSAT}{GapUNSAT}
\newlang{\OV}{OV}
\newlang{\PRIMES}{PRIMES}
\renewlang{\PCP}{PCP}
\newlang{\PCPP}{PCPP}
\newclass{\FMA}{FMA}
\newclass{\Avg}{Avg}
\newclass{\ZPEXP}{ZPEXP}
\newclass{\DLOGTIME}{DLOGTIME}
\newclass{\ALOGTIME}{ALOGTIME}
\newclass{\ATIME}{ATIME}%
\newclass{\SZKA}{SZKA}
\newclass{\Laconic}{Laconic\text{-}}
\newclass{\APEPP}{APEPP}
\newclass{\SAPEPP}{SAPEPP}
\newclass{\TFSigma}{TF\Sigma}
\newclass{\NTIMEGUESS}{NTIMEGUESS}
\newclass{\FZPP}{FZPP}
\newclass{\SearchNP}{SearchNP}
\newclass{\UEoPL}{UEoPL}
\newclass{\EoPL}{EoPL}
\newclass{\SoPL}{SoPL}
\newclass{\CLS}{CLS}
\newclass{\PWPP}{PWPP}

\newlang{\Formula}{Formula}
\newlang{\THR}{THR}
\newcommand{\MAJ}{\mathsf{MAJ}}
\newcommand{\DOR}{\mathsf{DOR}}
\newlang{\ETHR}{ETHR}
\newlang{\Midbit}{Midbit}
\newlang{\LCS}{LCS}
\newlang{\TAUT}{TAUT}

\newcommand{\dwPHP}{\mathrm{dwPHP}}

\newcommand{\Res}{\mathsf{Res}}

\newcommand{\PV}{\mathsf{PV}}
\newcommand{\APC}{\mathsf{APC}}

\newcommand{\AND}{\mathsf{AND}}

\newcommand{\XOR}{\mathsf{XOR}}

\newcommand{\calA}{\mathcal{A}}
\newcommand{\calB}{\mathcal{B}}

\newcommand{\calD}{\mathcal{D}}

\newcommand{\calG}{\mathcal{G}}
\newcommand{\calH}{\mathcal{H}}

\newcommand{\calL}{\mathcal{L}}

\newcommand{\calP}{\mathcal{P}}

\newcommand{\calU}{\mathcal{U}}
\newcommand{\calX}{\mathcal{X}}

\newcommand{\N}{\mathbb{N}}

\newcommand{\F}{\mathbb{F}}

\newcommand{\Eval}{\mathsf{Eval}}%

\newcommand{\Range}{\mathrm{Range}}

\newcommand{\Adv}{\mathsf{Adv}}

\newcommand{\TT}{\mathsf{TT}}

\newcommand{\Krajicek}{Kraj\'{\i}\v{c}ek\xspace}
\newcommand{\Jerabek}{Je\v{r}\'{a}bek\xspace}

\newcommand{\Prover}{\mathsf{Prover}}
\newcommand{\Verifier}{\mathsf{Verifier}}

\newcommand{\cbra}[1]{\left\{ #1 \right\}}

\newcommand{\iO}{i\mathcal{O}}

\usepackage[inline,marginclue,index]{fixme}  %
\fxsetup{theme=color,mode=multiuser}

\definecolor{color1}{RGB}{46,134,193}
\FXRegisterAuthor{hanlin}{ahanlin}{\colorbox{color1}{\color{white}Hanlin}}

\usepackage{aliascnt} %

\def\Anonymity{0}%

\begin{document}

\newgeometry{margin=0.85in}
\title{Hardness of Range Avoidance and\\Proof Complexity Generators from Demi-Bits}
\ifnum\Anonymity=0
\author{
	Hanlin Ren\\ \small{Institute for Advanced Study} \\ \small{\texttt{\href{mailto:h4n1in.r3n@gmail.com}{h4n1in.r3n@gmail.com}}}
	\and
	Yichuan Wang\\ \small{University of California, Berkeley} \\ \small{\texttt{\href{mailto:yichuan.tcs@gmail.com}{yichuan.tcs@gmail.com}}}
    \and
    Yan Zhong\\ \small{Johns Hopkins University} \\ \small{\texttt{\href{yanzhong.cs@gmail.com}{yanzhong.cs@gmail.com}}}
}
\fi

\maketitle
\begin{abstract}
Given a circuit $G: \{0, 1\}^n \to \{0, 1\}^m$ with $m > n$, the \emph{range avoidance} problem ($\textsc{Avoid}$) asks to output a string $y\in \{0, 1\}^m$ that is not in the range of $G$. Besides its profound connection to circuit complexity and explicit construction problems, this problem is also related to the existence of \emph{proof complexity generators} --- circuits $G: \{0, 1\}^n \to \{0, 1\}^m$ where $m > n$ but for every $y\in \{0, 1\}^m$, it is infeasible to prove the statement ``$y\not\in\Range(G)$'' in a given propositional proof system. %

This paper connects these two problems with the existence of \emph{demi-bits generators}, a fundamental cryptographic primitive against nondeterministic adversaries introduced by Rudich (RANDOM~'97). 
\begin{itemize}
    \item We show that the existence of demi-bits generators implies $\textsc{Avoid}$ is hard for nondeterministic algorithms. This resolves an open problem raised by Chen and Li (STOC~'24). Furthermore, assuming the demi-hardness of certain LPN-style generators or Goldreich's PRG, we prove the hardness of $\Avoid$ even when the instances are constant-degree polynomials over $\mathbb{F}_2$. 
    \item We show that the dual weak pigeonhole principle is unprovable in Cook's theory $\mathsf{PV}_1$ under the existence of demi-bits generators secure against $\mathbf{AM}/_{O(1)}$, thereby separating \Jerabek's theory $\mathsf{APC}_1$ from $\mathsf{PV}_1$. Previously, Ilango, Li, and Williams (STOC~'23) obtained the same separation under different (and arguably stronger) cryptographic assumptions.
    \item We transform demi-bits generators to proof complexity generators that are \emph{pseudo-surjective} in certain parameter regime. Pseudo-surjectivity is the strongest form of hardness considered in the literature for proof complexity generators.
\end{itemize}

Our constructions are inspired by the recent breakthroughs on the hardness of $\textsc{Avoid}$ by Ilango, Li, and Williams (STOC~'23) and Chen and Li (STOC~'24). We use \emph{randomness extractors} to significantly simplify the construction and the proof.
\end{abstract}

\section{Introduction}

This paper makes progress on the hardness of the \emph{range avoidance problem} and the existence of \emph{proof complexity generators}. We begin with a brief overview of these two lines of research.

\subsection{Range Avoidance}
The \emph{range avoidance} problem ($\Avoid$) is a total search problem introduced by~\cite{KKMP21,Korten21,RenSW22}. Given a Boolean circuit $G: \{0,1\}^n \rightarrow \{0,1\}^m$ with $m > n$ (usually we also require $m\leq\poly(n)$), the goal is to find a string $y \in \{0,1\}^m$ such that $y \not\in \Range(G)$. This problem has attracted considerable interest due to its connection to central problems in complexity theory such as explicit constructions~\cite{Korten21, RenSW22, ChenHLR23, GuruswamiLW22, GajulapalliGNS23} and circuit lower bounds~\cite{Korten21, CHR24, Li24, KortenP24}. We refer the reader to~\cite{Korten-EATCS} for a comprehensive survey on the range avoidance problem. %

The range avoidance problem admits a trivial randomized algorithm: simply output a uniformly random $m$-bit string, which will lie outside the range of $G$ with high probability. On the other hand, deterministic algorithms for $\Avoid$ would imply breakthroughs in explicit constructions and circuit lower bounds~\cite{Korten21, RenSW22, GuruswamiLW22, GajulapalliGNS23}. Since such breakthroughs are widely believed to be \emph{true} (albeit \emph{difficult to prove}), the aforementioned results only suggest that deterministic algorithms for $\Avoid$ would be difficult to obtain \emph{unconditionally}, rather than that such algorithms are \emph{unlikely to exist}. This raises a natural question: Is there a deterministic algorithm for $\Avoid$?

Perhaps surprisingly, recent results suggested that the answer is likely \emph{no} under plausible cryptographic assumptions. Ilango, Li, and Williams~\cite{DBLP:conf/stoc/IlangoLW23} showed that $\Avoid$ is hard for deterministic algorithms assuming the existence of subexponentially secure indistinguishability obfuscation ($\iO$) and that $\NP\neq \coNP$. Chen and Li~\cite{DBLP:conf/stoc/ChenL24} extended this result and showed that $\Avoid$ is hard even for \emph{nondeterministic} algorithms, under certain assumptions regarding the nondeterministic hardness of LWE (Learning with Errors) or LPN (Learning Parity with Noise). In addition to providing compelling evidence for the hardness of $\Avoid$, these results establish a strong separation between deterministic and randomized algorithms (recall that there exists a trivial randomized algorithm for $\Avoid$). %

The hardness results in~\cite{DBLP:conf/stoc/IlangoLW23, DBLP:conf/stoc/ChenL24} open up several exciting research directions:
\begin{enumerate}
    \item {\bf Can the hardness of $\Avoid$ be based on weaker (or alternative) assumptions?}
    
    The assumptions used in prior work, i.e., $\iO$~\cite{DBLP:conf/stoc/IlangoLW23} and public-key encryption~\cite{DBLP:conf/stoc/ChenL24}, belong to Cryptomania in the terminology of Impagliazzo's worlds~\cite{Impagliazzo95}. Can we base the hardness of range avoidance on assumptions of a ``Minicrypt'' flavor, such as one-way functions or pseudorandom generators? Additionally, both~\cite{DBLP:conf/stoc/IlangoLW23} and~\cite{DBLP:conf/stoc/ChenL24} rely on \emph{subexponential} indistinguishability assumptions\footnote{More precisely, the assumptions in \cite{DBLP:conf/stoc/IlangoLW23, DBLP:conf/stoc/ChenL24} assert subexponential indistinguishability against polynomial-time adversaries. This level of security is referred to as ``JLS-security'' in~\cite{DBLP:conf/stoc/IlangoLW23}, where ``JLS'' comes from the strengths of the ``well-founded'' assumptions used to construct $\iO$ in~\cite{JainLS21}.}. Are such subexponential assumptions necessary?
    \item {\bf Can we obtain hardness of $\Avoid$ for instances computed by restricted circuits?}
    
    Previously, under assumptions related to LWE, Chen and Li~\cite{DBLP:conf/stoc/ChenL24} showed that $\Avoid$ remains hard even when each output bit of $G$ is computed by a so-called ``$\DOR\circ\MAJ\circ\AND_{O(\log n)}$ circuit''. No such results were known for other restricted circuit classes. The related \emph{remote point} problem has been shown to be hard under LPN-style assumptions for $\XOR\circ\AND_{O(\log n)}$ circuits (i.e., $O(\log n)$-degree polynomials over $\F_2$)~\cite{DBLP:conf/stoc/ChenL24}.
\end{enumerate}

This paper makes progress on both fronts. We show that $\Avoid$ is hard for nondeterministic algorithms under the existence of \emph{demi-bits generators} with sufficient stretch\footnote{The stretchability of generic demi-bits generators is only partially understood. Recent work of Tzameret and Zhang~\cite{tzameret_et_al:LIPIcs.ITCS.2024.95} shows that demi-bits generator with $1$-bit stretch $G : \{0, 1\}^n \to \{0, 1\}^{n+1}$ implies those with a \emph{sublinear} bits of stretch $G': \{0, 1\}^n \to \{0, 1\}^{n + n^c}$ for any constant $0 < c < 1$. This is the first proof that generic demi-bits generators are stretchable at all, but it still falls short of the \emph{linear} or \emph{polynomial} stretch assumed in our hypothesis.}%
~\cite{Rudich97}, e.g.\ generators $G : \{0,1\}^n \to \{0,1\}^{10n}$. A formal definition of demi-bits generators is deferred to~\autoref{sec: def of demi-bits}, and candidate constructions supporting their existence are discussed in~\autoref{sec: candidate demi-bits}. For the purpose of this introduction, it suffices to keep in mind that demi-bits generators are a version of cryptographic pseudorandom generators secure against nondeterministic adversaries.

We highlight three key features of our results here:
\begin{enumerate}
    \item {\bf Minicrypt-style assumptions against nondeterministic adversaries.}
    
    Roughly speaking, demi-bits generators are (cryptographic) pseudorandom generators secure against nondeterministic adversaries\footnote{In fact, Rudich~\cite{Rudich97} introduced two ways to define pseudorandomness against nondeterministic adversaries: super-bits and demi-bits. Demi-bits are weaker than super-bits.}. They are arguably a natural ``Minicrypt'' analog of pseudorandom generators in the context of cryptography against nondeterministic adversaries. Moreover, our results only rely on the \emph{super-polynomial} hardness of these demi-bits generators, thereby completely getting rid of the subexponential (or ``JLS''-style) assumptions used in prior work.
    \item {\bf Hardness for restricted circuit classes.}
    
    Under the assumption that certain concrete demi-bits generators are secure (e.g., those based on LPN or Goldreich's PRG), we show that the range avoidance problem remains hard for nondeterministic algorithms even when the underlying circuits belong to $\XOR\circ\AND_{O(1)}$, i.e., constant-degree polynomials over $\mathbb{F}_2$.

    \item {\bf Simplicity of the proof for hardness of range avoidance.}
    
    In contrast to~\cite{DBLP:conf/stoc/IlangoLW23,DBLP:conf/stoc/ChenL24}, which rely on sophisticated and delicate adaptations of high-end cryptographic assumptions, our proof of the hardness of range avoidance (\autoref{thm: demi-bit to Avoid notin FP}) is based solely on the most elementary pseudorandom constructions together with the existence of demi-bits. This approach distills the arguments of~\cite{DBLP:conf/stoc/IlangoLW23,DBLP:conf/stoc/ChenL24} to their essence, yielding a proof that is both conceptually cleaner and technically simpler. In fact, the proof of~\autoref{thm: demi-bit to Avoid notin FP} fits in just half a page. Moreover, by isolating the core ingredients, this simplified framework opens the door to potential extensions and generalizations that may be harder to see in the more cryptographically heavy approaches.
    
\end{enumerate}

\subsection{Proof Complexity Generators}

Let $G: \{0, 1\}^n \to \{0, 1\}^m$ be a Boolean circuit where $m > n$, and $\calP$ be a propositional proof system. We say that $G$ is a (secure) \emph{proof complexity generator}~\cite{ABRW04, krajivcek2001weak} against $\calP$ if, for every string $y \in \{0, 1\}^m$, the (properly encoded) statement ``$y \not\in \Range(G)$'' does not admit short proofs in $\calP$.\footnote{If $y$ is in fact in the range of $G$, then ``$y\not\in\Range(G)$'' is a false statement and hence has no proof in any sound proof system. Therefore, this requirement is equivalent to that, for every $y\not\in\Range(G)$, the tautology ``$y\not\in\Range(G)$'' is hard to prove in $\calP$.} A comprehensive survey about proof complexity generators can be found in~\cite{Krajicek-generator-book}.

The study of proof complexity generators is motivated by at least the following themes:
\begin{enumerate}
    \item {\bf Pseudorandomness in proof complexity~\cite{ABRW04}.}
    
    A standard \emph{pseudorandom generator} (PRG) $G: \{0, 1\}^n \to \{0, 1\}^m$~\cite{Yao82a} fools a (polynomial-time) algorithm $\calD$ if $\calD$ cannot distinguish the outputs of $G$ from truly random $m$-bit strings; that is, $\calD(\calU_m) \approx \calD(G(\calU_n))$, where $\calU_\ell$ denotes the uniform distribution over $\ell$-bit strings. Analogously, one can say that $G$ fools a propositional proof system $\calP$ if $\calP$ cannot distinguish between the outputs of $G$ and truly random $m$-bit strings, and a natural way of formalizing this is to say that $\calP$ cannot efficiently prove any string outside the range of $G$.

    Following the idea of pseudorandomness in proof complexity, subsequent works~\cite{Krajicek04, Pich11, Krajicek11, razborov2015pseudorandom, Khaniki22} studied the hardness of the \emph{Nisan--Wigderson} generator (\cite{NisanW94}) as a proof complexity generator in various settings. An influential conjecture of Razborov asserts that the Nisan--Wigderson generator based on any ``sufficiently hard'' function in $\NP\cap\coNP$ is a proof complexity generator against Extended Frege~\cite[Conjecture 2]{razborov2015pseudorandom}; that is, computational hardness can be transformed into proof complexity pseudorandomness.

    \item {\bf Candidate hard tautologies for strong proof systems.}
    
    There are two difficulties in proving lower bounds for strong proof systems such as Frege and Extended Frege: the lack of techniques and the lack of candidate hard tautologies. The latter problem was highlighted by Bonet, Buss, and Pitassi~\cite{BBP95}, who demonstrated that many combinatorial tautologies can be proved efficiently in Frege, hence disqualifying them as hard candidates. This issue has been further discussed in~\cite{krajivcek2001tautologies, Krajicek04, SanthanamT21, Khaniki22}.
    
    Tautologies from proof complexity generators are among the few natural candidates that appear hard for strong proof systems. It seems plausible that for some mapping $G: \{0, 1\}^n \to \{0, 1\}^{n+1}$ and some (or even \emph{every}) $y\in\{0, 1\}^{n+1}\setminus \Range(G)$, the natural CNF encoding of the tautology ``$y\not\in\Range(G)$'' requires super-polynomially long Extended Frege proofs.

    \item {\bf Unprovability of circuit lower bounds.}
    
    Given our very limited progress in circuit complexity, it is tempting to conjecture that circuit lower bounds are hard to prove in formal proof systems. For a Boolean function $f: \{0, 1\}^n \to \{0, 1\}$ and a size parameter $s$, one can write down a propositional formula ${\sf lb}(f, s)$ (of size $2^{O(n)}$) asserting that no circuit of size at most $s$ computes $f$. The proof complexity of such formulas have been studied extensively~\cite{Razborov98, Razborov04, Krajicek04, razborov2015pseudorandom, Pich15, PichS19, SanthanamT21, PichS22}, due to its implications for the metamathematics of complexity theory.

    Consider the \emph{truth table generator} $\TT: \{0,1\}^{\poly(s)} \to \{0,1\}^{2^n}$, which maps a size-$s$ circuit $C$ to its $2^n$-bit truth table. By definition, $\TT$ is a proof complexity generator against a proof system $\calP$ if and only if $\calP$ cannot efficiently prove any circuit lower bound ${\sf lb}(f, s)$. \Krajicek~\cite{Krajicek04} introduced the notion of \emph{pseudo-surjectivity} and showed that $\TT$ is the hardest pseudo-surjective generator: The existence of any generator pseudo-surjective against $\calP$ implies that $\TT$ is pseudo-surjective against $\calP$ (and thus that $\calP$ cannot prove circuit lower bounds). Razborov~\cite{razborov2015pseudorandom} further showed unprovability of circuit lower bounds in the proof system ${\rm Res}(\eps \log\log N)$ by exhibiting a proof complexity generator that is iterable for this system.\footnote{Iterability is a weaker notion than pseudo-surjectivity. \Krajicek~\cite{Krajicek04} also showed that $\TT$ is the ``hardest'' iterable generator.}
\end{enumerate}

\Krajicek~\cite{Krajicek04, Krajicek23, krajivcek2024existence} conjectured that there exists a proof complexity generator that is secure against every propositional proof system. One could also consider a slightly weaker conjecture that for every propositional proof system $\calP$, there is a proof complexity generator $C_\calP$ (possibly depending on $\calP$) that is hard against $\calP$. At first glance, these conjectures may appear unrelated to standard hardness assumptions in complexity theory or cryptography, as proof complexity generators require ``$y\not\in\Range(G)$'' to be hard to prove for \emph{every} $y$ (i.e., the \emph{best-case} $y$), while complexity-theoretic or cryptographic hardness assumptions tend to be either \emph{worst-case} or \emph{average-case}. We elaborate on the notion of ``best-case'' proof complexity in~\autoref{sec: perspective on best-case complexity}.

In this paper, we give strong evidence for the weaker conjecture by showing that it follows from the existence of demi-bits generators (with sufficiently large stretch)~\cite{Rudich97}. The latter is a natural and fundamental conjecture in the study of \emph{cryptography against nondeterministic adversaries}. Furthermore, we show that our generators are even pseudo-surjective under certain regimes.\footnote{The parameters of our pseudo-surjectivity results \emph{fall just short} of those required to apply \Krajicek's result~\cite{Krajicek04}, hence they do not imply the hardness of the truth table generator. This limitation is inherent; we discuss this issue in more details after presenting \autoref{thm: main pseudo-pcg}.} %

\subsection{Our Results}
\paragraph{Hardness of range avoidance.} Our main result is that the existence of demi-bits generators implies that $\Avoid\not\in\SearchNP$, i.e., $\Avoid$ is hard for nondeterministic search algorithms. %

\begin{theorem}[Main]\label{thm: main demi-bit to Avoid notin FP}
	If there exists a demi-bits generator $G: \{0, 1\}^n \to \{0, 1\}^{10n}$, then $\Avoid \notin \SearchNP$. %
\end{theorem}

In fact, we show that composing the demi-bits generator with a hash function in some pairwise independent hash family would yield a hard instance for $\Avoid$. In its full generality, our arguments hold for arbitrary \emph{strong seeded extractors}, and the theorem below follows from the leftover hash lemma~\cite{ILL89}, which guarantees that pairwise independent hash families are such extractors; see \autoref{thm: demi-bit to Avoid notin FP} for details. %

We present the version using pairwise independent hash families here due to its elegance:

\begin{restatable}{theorem}{ThmDemiBitsPlusHash}\label{thm: intro demi-bits plus hash}
    Let $G: \{0, 1\}^n \to \{0, 1\}^N$ be a demi-bits generator, $\calH = \{h: \{0, 1\}^N \to \{0, 1\}^m\}$ be a family of pairwise independent hash functions, and $\calA$ be a nondeterministic polynomial-time algorithm. If $N > 10m$ and $m > n$, then there exists $h\in \calH$ such that $\calA$ fails to solve the range avoidance problem on the input $h\circ G$.
\end{restatable}

As discussed before, this result improves upon~\cite{DBLP:conf/stoc/IlangoLW23, DBLP:conf/stoc/ChenL24} in several key aspects. First, we only require super-polynomial hardness of the demi-bits generators, thereby completely eliminating the subexponential- or JLS-hardness assumptions. Second, our assumptions are solely based on the existence of demi-bits generators, a primitive arguably situated within ``nondeterministic Minicrypt.'' 
Finally, by instantiating the extractors with pairwise independent hash functions computable by linear transformations over $\F_2$ and using demi-bits generators computable by constant-degree $\F_2$-polynomials, we establish the hardness of $\Avoid$ even for circuits where each output bit is computable in constant $\F_2$-degree (i.e., $\XOR\circ\AND_{O(1)}$ circuits):

\begin{corollary}[Informal]\label{cor: main degree-O(1) Avoid notin FP}
	Assuming the existence of demi-bits generators computable in $\XOR\circ\AND_{O(1)}$ (\autoref{assumption: demi-bits in degree 2}), the range avoidance problem for $\XOR\circ\AND_{O(1)}$ circuits is not in $\SearchNP$.
\end{corollary}

\paragraph{Proof complexity generators.} Building on this result, we show that for any fixed propositional proof system $\calP$ closed under certain reductions, demi-bits generators for $\calP$ imply proof complexity generators for $\calP$. In particular, the existence of demi-bits generators secure against $\NP/_\poly$ implies the weaker version of \Krajicek's conjecture, providing strong evidence that the latter conjecture is true. 

Moreover, this result suggests a new approach for constructing proof complexity generators for concrete proof systems closed under certain reductions, such as $\Res[\oplus]$: it suffices to construct demi-bits generators secure against \emph{the same proof system}.

\begin{theorem}\label{thm: main pcg}
    Let $\calP$ be a proof system closed under ``simple parity reductions'' (\autoref{def: parity reduction}). If there exists a demi-bits generator $G: \{0, 1\}^n \to \{0, 1\}^{10n}$ secure against $\calP$, then there is a (non-uniform) proof complexity generator secure against $\calP$.
\end{theorem}

\paragraph{Unprovability of $\dwPHP(\PV)$ in $\PV$ from demi-bits.} A central goal in bounded arithmetic is to delineate the logical power required to formalize reasoning about computational complexity. Two well-studied theories in this context are Cook’s theory $\PV_1$~\cite{Cook75}, which corresponds to reasoning in deterministic polynomial time, and Jeřábek’s theory $\APC_1$~\cite{Jerabek04, Jerabek07}\footnote{Note that the terminology $\APC_1$ was first used in~\cite{BussKT14}.}, which extends $\PV_1$ by adding the \emph{Dual Weak Pigeonhole Principle} for polynomial-time functions ($\dwPHP(\PV)$), and captures aspects of randomized polynomial-time reasoning.

Despite decades of interest, it has remained open whether $\APC_1$ and $\PV_1$ are actually distinct theories, i.e., whether $\dwPHP(\PV)$ is unprovable in $\PV_1$. Recently, Ilango, Li, and Williams~\cite{DBLP:conf/stoc/IlangoLW23} provided the first evidence separating the two: they showed that $\dwPHP(\PV)$ is unprovable in $\PV_1$ under the assumptions that indistinguishability obfuscation ($\iO$) with JLS-security exists and that $\coNP$ is not infinitely often in $\AM$. We remark that the same separation was also shown by \Krajicek~\cite{Krajicek21}, albeit under an assumption that is regarded as ``unlikely'' ($\P$ admits fixed-polynomial size circuits). In this work, we establish the same separation assuming the existence of demi-bits generators against $\AM/_{O(1)}$.

\begin{restatable}{theorem}{ThmPVvsAPC}\label{thm: main separation}
    Assuming there exists a demi-bits generator $G: \{0,1\}^n \to \{0,1\}^{\omega(n)}$ secure against $\AM/_{O(1)}$, the Dual Weak Pigeonhole Principle for polynomial-time functions $(\dwPHP(\PV))$ is not provable in $\PV$. (In particular, $\APC_1$ is a strict extension of $\PV_1$.)
\end{restatable}

The only property of $\PV_1$ used in our argument is the KPT witnessing theorem~\cite{KRAJICEK1991143}, which states that if $\PV_1$ proves the dual weak pigeonhole principle for polynomial-time functions, then there exists a deterministic polynomial-time Student that wins the Student-Teacher game for solving $\Avoid$ in $O(1)$ rounds (see~\autoref{sec: bounded arithmetic prelim}). Our separation result in \autoref{thm: main separation} follows by showing that no such Student algorithm exists.%
Moreover, for any parameter $k = k(n)$, assuming the existence of demi-bits generators secure against $\AM/_{O(\log k)}$, we further rule out deterministic polynomial-time Students that wins the Student-Teacher game within $k$ rounds.%

\begin{theorem}\label{thm: main no-avoid^O(1)}
    Let $m = m(n) > n$ and $k = k(n)$ be parameters. If there exists a demi-bits generator $G: \{0,1\}^n \to \{0,1\}^{100km}$ secure against $\AM/_{O(\log k)}$, then there is no polynomial-time deterministic algorithm for $\Avoid$ on circuits with $n$ inputs and $m$ outputs using $k$ circuit-inversion oracle queries.
\end{theorem}

\paragraph{Pseudo-surjective generators.} Assuming demi-bits generators against $\NP/_{\poly}$, we show that our proof complexity generators are even pseudo-surjective against every proof system. (The precise definition of \emph{$k$-round} pseudo-surjectivity is presented in \autoref{def:pseudo-surjectivity}.)

\begin{theorem}\label{thm: main pseudo-pcg}
    Let $m = m(n) > n$ and $k = k(n)$ be parameters. If there exists a demi-bits generator $G:\{0,1\}^n\to\{0,1\}^{100km}$ secure against $\NP/_{\poly}$, then for every non-uniform propositional proof system $\calP$, there is a non-uniform proof complexity generator $G_{\calP, k}: \{0, 1\}^n \to \{0, 1\}^m$ that is $k$-round pseudo-surjective against $\calP$.
\end{theorem}

\Krajicek~\cite{Krajicek04} proved that, under appropriate parameter settings, the existence of pseudo-surjective proof complexity generators is equivalent to the pseudo-surjectivity of the truth table generator. As a corollary, the existence of pseudo-surjective generators against a proof system $\calP$ implies that \emph{every} circuit lower bound is hard to prove in $\calP$. 

However, the parameters in our \autoref{thm: main pseudo-pcg} \emph{fall short} of applying \Krajicek's result and therefore do \emph{not} imply the pseudo-surjectivity of the truth table generator. Specifically, to apply \Krajicek's results, we need a proof complexity generator that is computable by circuits of size $s$ and is $k$-round pseudo-surjective for some $k \gg s$; see the proof of~\cite[Theorem 4.2]{Krajicek04} for detailed discussions. In contrast, \autoref{thm: main pseudo-pcg} guarantees a generator computable by circuits of size $\poly(n, k)$ that is $k$-round pseudo-surjective, which is in the regime where $k<s$ and thus outside the reach of \Krajicek's equivalence.

This limitation is inherent to the generality of our result: we construct generators secure against \emph{all} proof systems, whereas under the assumption $\E\not\subseteq \SIZE[2^{o(n)}]$, there exists a proof system that \emph{can} prove circuit lower bounds (e.g.,~by simply hardwiring an axiom that certain $\E$-complete language has exponential circuit complexity). Thus, under this standard hardness assumption, it is provably impossible to extend our results to the regime where $k \gg s$ and thereby obtain pseudo-surjectivity of the truth table generator. It remains an intriguing open question whether our approach can be refined to construct proof complexity generators of size $s$ that are $k$-round pseudo-surjective against specific systems such as Extended Frege, for some $k \gg s$. Such a result would imply that Extended Frege cannot prove \emph{any} circuit lower bounds.

Finally, we comment on the strength of the adversaries required in our assumptions on demi-bits generators. In the main theorem (\autoref{thm: main demi-bit to Avoid notin FP}), a $\SearchNP$ algorithm for $\Avoid$ is transformed into a nondeterministic adversary that breaks the demi-bits generator. Since $\SearchNP$ is a uniform class, it suffices to assume that the generator is secure against uniform nondeterministic adversaries. In contrast, \autoref{thm: main no-avoid^O(1)} requires the generator to be secure against $\AM/_{O(\log k(n))}$ adversaries. This is because the adversary invokes the Goldwasser--Sipser protocol~\cite{GS86} which is in $\AM$, and needs to hardwire the index of the circuit-inversion query that succeeds with good probability. In \autoref{thm: main pseudo-pcg}, we require security against $\NP/_\poly$ adversaries, as our adversary needs to hardwire a ``good'' sequence of teacher responses in the student-teacher game, thus is highly non-uniform.

\subsection{Perspective: Average-Case to Best-Case Reductions in Proof Complexity}\label{sec: perspective on best-case complexity}

Theoretical computer science has traditionally focused on the \emph{worst-case} complexity of problems: An algorithm $\calA$ solves a problem if $\calA(x)$ succeeds on every input $x$. Motivated by practical heuristics (where worst-case analysis tends to be overly pessimistic) and cryptography (where worst-case hardness is not sufficient for security), \emph{average-case} complexity has emerged as an important research direction~\cite{Impagliazzo95, BogdanovT06}. In this setting, fixing a distribution $\calD$ over inputs, an algorithm $\calA$ solves a problem if $\calA(x)$ succeeds with good probability over $x\sim \calD$. Recently, average-case complexity has received attention in proof complexity as well: For example, \cite{Pang21,DSPR23,CDSNPR23} proved proof complexity lower bounds for \textsc{Clique} and \textsc{Coloring} for \emph{random} graphs.

An important topic in average-case complexity is \emph{worst-case to average-case} reductions: reductions showing that if a problem $L$ is hard in the worst-case then a related problem $L'$ is hard on average. Worst-case to average-case reductions are known for the Permanent~\cite{CaiPS99}, Discrete Logarithm~\cite{BlumM84}, Quadratic Residuosity~\cite{GoldwasserM84}, certain lattice problems~\cite{Ajtai96}, and more recently for problems in meta-complexity~\cite{Hirahara18}. On the other hand, ``black-box'' worst-case to average-case reductions are unlikely to exist for $\NP$-complete problems~\cite{FeigenbaumF93, BogdanovT06-reduction}.

In contrast, the notion of \emph{best-case} complexity has received far less attention. Perhaps one reason is that this notion is often trivial in standard computational complexity: for every language $L$, either the all-zero function or the all-one function can decide $L$ on the ``best'' input.\footnote{One notable exception appears in recent derandomization results~\cite{CT21-FOCS} based on \emph{almost-all-input} hardness assumptions. In particular, it was shown that $\mathrm{pr}\mathbf{P} = \mathrm{pr}\mathbf{BPP}$ follows from the existence of depth-efficient multi-output functions with high best-case complexity against randomized algorithms.} However, in proof complexity, best-case hardness is a meaningful and natural notion, as illustrated by proof complexity generators, which are stretching functions $G$ such that the statement ``$y\not\in\Range(G)$'' is hard to prove even for the \emph{best} choice of $y$.

In this context, our results can be interpreted as an \emph{average-case to best-case} reduction in proof complexity. Indeed, \autoref{thm: main pcg} transforms demi-bits generators, where statements of the form ``$y\not\in\Range(G)$'' is hard to prove for an \emph{average-case} $y$, to proof complexity generators, where such statements are hard for a \emph{best-case} $y$.

We find the existence of such average-case to best-case reductions quite surprising. Our arguments crucially exploit the power of nondeterministic computation, and the phenomenon of average-case to best-case reductions seems unique to the setting of proof complexity and hardness against nondeterministic algorithms. We believe that further exploring the scope and limitations of average-case to best-case reductions is a promising direction for future research.

\def\rank{\mathrm{rank}}
We remark that there are also worst-case to best-case reductions in proof complexity. \Krajicek~\cite{krajicek2007proof} constructed a proof complexity generator whose hardness can be based on the hardness of the pigeonhole principle, thereby reducing the best-case hardness of an entire family of tautologies to that of a single tautology. Inspired by this example, Garlik, Gryaznov, Ren, and Tzameret~\cite{GGRT25} recently showed a worst-case to best-case reduction for the \emph{rank principles}. Let ``$\rank(A) > r$'' denote the collection of polynomial equations expressing that the rank of an $n\times n$ matrix $A$ is greater than $r$. If a proof system (closed under certain algebraic reductions) cannot prove ``$\rank(I_n) > r$'' where $I_n$ is the $n\times n$ identity matrix, then it also cannot prove ``$\rank(A) > r$'' for \emph{every} $n\times n$ matrix $A$.

\subsection{Concurrent Works}
In an independent and concurrent work, Ilango~\cite{Ilango25} presented a different proof of~\autoref{thm: main demi-bit to Avoid notin FP} and \autoref{thm: main pcg}. Ilango's proof implies the \emph{average-case} hardness of $\Avoid$ under a natural distribution (namely, it is hard for errorless nondeterministic heuristic algorithms to output a truth table with high $O$-oracle circuit complexity given the truth table of a random oracle $O$). Based on this result, Ilango constructed a (non-uniform) non-interactive proof system that ``looks'' zero knowledge to every proof system. However, Ilango did not show the pseudo-surjectivity of his generators. We discuss Ilango's proof in more detail in \autoref{sec: RAHUL}.

\section{Preliminaries}\label{sec: preliminaries}

\subsection{Demi-Bits Generators}\label{sec: def of demi-bits}
\begin{definition}[Demi-Bits Generators]
    Let $n,m$ be length parameters such that $n<m$. A function $G:\{0,1\}^n\to \{0,1\}^m$ is an \emph{$(s,\eps)$-secure demi-bits generator} if there is no $\NP/_\poly$ adversary $\Adv$ of size $s$ such that 
    \begin{align*}
        \Pr_{y\sim\{0,1\}^m}[\Adv(y)=1]\ge\eps \quad\text{and}\quad\Pr_{x\sim\{0,1\}^n}[\Adv(G(x))=1]=0.
    \end{align*}
\end{definition}

This paper requires demi-bits generators with large stretch and computable with small circuit complexity. In particular, we need the following assumptions:
\begin{restatable}[{Demi-Bits Generators with Polynomial Stretch}]{assumption}{PolyStretchDemiBits}\label{assumption: demi-bits with poly stretch}
    For every constant $c\ge 1$, there exists a family of demi-bits generators $\{g_n: \{0, 1\}^n \to \{0, 1\}^{n^c}\}$ secure against $\NP/_\poly$.
\end{restatable}

\begin{restatable}[{Demi-Bits Generators with $n^{1+\eps}$ Stretch in Constant Degree}]{assumption}{DegreeTwoDemiBits}\label{assumption: demi-bits in degree 2}
    There exist constants $\eps > 0$, $d\ge 2$, and a (non-uniformly computable) family of demi-bits generators $\{g_n: \{0, 1\}^n \to \{0, 1\}^{n^{1+\eps}}\}$ secure against $\NP/_\poly$, such that each output bit of $g_n$ is computable by a degree-$d$ polynomial over $\F_2$ (i.e., an $\XOR\circ\AND_d$ circuit).
\end{restatable}

Our main hardness results for $\Avoid$ will be based on \autoref{assumption: demi-bits with poly stretch}; we also need \autoref{assumption: demi-bits in degree 2} to obtain hardness results for constant-degree $\Avoid$. In \autoref{sec: candidate demi-bits}, we justify these assumptions and provide some candidate constructions.

\subsection{Arthur--Merlin Protocols}

An \emph{Arthur--Merlin} protocol~\cite{Babai85} for a language $L$ is a constant-round public-coin interactive protocol between a computationally unbounded $\Prover$ (Merlin) and a randomized polynomial-time $\Verifier$ (Arthur) that satisfies the following properties for every input $x$:
\begin{itemize}
    \item {\bf Completeness:} If $x\in L$, then there is a $\Prover$ that makes the $\Verifier$ accept w.p.~$\ge 2/3$.
    \item {\bf Soundness:} If $x\not\in L$, then no $\Prover$ can make the $\Verifier$ accept w.p.~$>1/3$.
\end{itemize}

Let $\AM$ denote the set of languages with an Arthur--Merlin protocol. The round-collapse theorem of Babai~\cite{Babai85} implies that every language in $\AM$ actually has an Arthur--Merlin protocol with two rounds: $\Verifier$ sends the first message, $\Prover$ sends a proof, and $\Verifier$ decides whether $x\in L$. Hence, $\AM$ can be seen as a randomized version of $\NP$; indeed, one can prove that $\AM = \NP$ under circuit lower bound assumptions~\cite{KlivansM02, MiltersenV05, ShaltielU05, ShaltielU06}.

\paragraph{Goldwasser--Sipser set lower bound protocol.} We need the following well-known $\AM$ protocol for proving lower bounds on the size of efficiently recognizable sets.
\begin{lemma}[\cite{GS86}, also see~{\cite[section 8.4]{AB09}}]\label{lemma:GS-protocol}
There is an Arthur--Merlin protocol such that the following holds. Suppose that both $\Prover$ and $\Verifier$ receive a nondeterministic circuit $C:\{0,1\}^n \to \{0,1\}$ and a number $s\le 2^n$. Let $S=\cbra{x\in \{0,1\}^n : C(x)=1}$.
\begin{itemize}
    \item \textbf{Specification:} The protocol is a two-round public-coin protocol, in which the $\Verifier$ first sends a random string $r$ and receives a message $m$, then deterministically decides whether to accept based on $r$ and $m$.
    \item \textbf{Completeness: }If $|S|\ge s$, then w.p.~$\ge 2/3$ over $r$, there exists a proof $m$ that makes the $\Verifier$ accept.
    \item \textbf{Soundness: }If $|S|\le s/2$, then w.p.~$\ge 2/3$ over $r$, the $\Verifier$ rejects regardless of $m$.
\end{itemize}
\end{lemma}

\paragraph{Arthur--Merlin protocols as adversaries.}

\begin{definition}[Breaking demi-bits generators by $\AM$ adversaries]
    Let $m > n$. An $\AM$ adversary breaks a demi-bits generator $G:\{0, 1\}^n \to \{0, 1\}^m$ if both $\Prover$ and $\Verifier$ receives a common input $y\in\{0, 1\}^m$, and:
    \begin{itemize}
        \item for $\ge 1/3$ fraction of $y\in \{0, 1\}^m$, there exists a $\Prover$ that makes the $\Verifier$ accept with probability $\ge 2/3$;
        \item for all $y\in \Range(G)$, for every $\Prover$, the $\Verifier$ accepts with probability $\le 1/3$.
    \end{itemize}
\end{definition}

We also consider $\AM$ adversaries with advice:

\begin{definition}[$\AM/_{k(n)}$ adversaries]
    An $\AM/_{k(n)}$ adversary is an Arthur--Merlin protocol where the $\Verifier$ is a probabilistic polynomial-time machine that additionally receives a $k(n)$-bit advice string $a_n$ (which may depend on the input length $n$ but not on the specific input $y$). The interaction on input $y$ proceeds as follows:
    \begin{compactenum}
        \item The $\Verifier$ uses $a_n$ and randomness $r$ to generate a message to the $\Prover$;
        \item The $\Prover$ replies with a message;
        \item The $\Verifier$ accepts or rejects based on $y$, $a_n$, $r$, and the $\Prover$'s response.
    \end{compactenum}
    The acceptance probabilities are still defined over $\Verifier$'s internal randomness, with the advice string fixed to $a_n$.
\end{definition}

\begin{proposition}
    Let $G$ be a demi-bits generator.
    \begin{itemize}
        \item If there exists an $\AM$ adversary breaking $G$, then there exists an $\NP/_{\poly}$ adversary breaking $G$ \textnormal{(\cite{Adleman78})}.
        \item For every constant $k\ge 2$, if there exists a $k$-round $\AM$ adversary breaking $G$, then there exists a (standard) two-round $\AM$ adversary breaking $G$ \textnormal{(\cite{Babai85})}.
    \end{itemize}
\end{proposition}

\subsection{\texorpdfstring{$\FNP$}{FNP} v.s.~\texorpdfstring{$\SearchNP$}{SearchNP}}
In this paper, we will distinguish between the two notions $\FNP$ and $\SearchNP$.

\begin{definition}[$\SearchNP$~\cite{DBLP:conf/stoc/ChenL24}]
    Let $P$ be a search problem and $R$ be the binary relation defining $P$. We say $P$ can be solved by a nondeterministic polynomial-time algorithm if there is a nondeterministic Turing machine $M$ such that for every input $x$,
    \begin{itemize}
        \item If $x$ has a solution, then $M(x)$ has an accepting computation path, and every accepting path will output a valid solution $y$, i.e., $R(x,y)$ is true.
        \item If $x$ has no solution, then $M(x)$ has no accepting computation path.
    \end{itemize}
    The class of search problems solvable by nondeterministic polynomial-time algorithm is defined as $\SearchNP$.
\end{definition}

\begin{definition}[$\FNP$~\cite{DBLP:conf/stoc/ChenL24}]
    The class of search problems defined by a polynomial-time relation, i.e.,~$R\in \mathbf{P}$ is defined as $\FNP$.
\end{definition}

While it is clear that $\FNP\subseteq \SearchNP$, the following example suggests that this inclusion is strict.

\begin{proposition}[\cite{DBLP:conf/stoc/ChenL24}]
    If $\mathbf{P}\neq \mathbf{NP}$, then there is a total search problem in $\SearchNP\setminus \FNP$.
\end{proposition}

For more knowledge about nondeterministic algorithms, readers are referred to~\cite[Section 2.4]{DBLP:conf/stoc/ChenL24}.

\subsection{Proof Complexity}\label{sec: proof complexity}

Recall that $\TAUT$, the set of DNF tautologies, is the canonical $\coNP$-complete problem. A \emph{propositional proof system} is simply a nondeterministic algorithm for $\TAUT$. More formally:

\begin{definition}[{\cite{CookR79}}]
    An algorithm $\calP(\varphi, \pi)$ is called a \emph{propositional proof system} if it satisfies the following conditions:\begin{itemize}[align=left]
        \item {(\bf Completeness)} For every $\varphi \in \TAUT$, there exists a string $\pi\in\{0, 1\}^*$ such that $\calP(\varphi, \pi)$ accepts.
        \item {(\bf Soundness)} For every $\varphi, \pi \in \{0, 1\}^*$, if $\calP(\varphi, \pi)$ accepts, then $\varphi \in \TAUT$.
        \item {(\bf Efficiency)} $\calP(\varphi, \pi)$ runs in deterministic $\poly(|\varphi| + |\pi|)$ time.
    \end{itemize}
    We say that $\calP$ is a \emph{non-uniform} propositional proof system if $\calP$ is a polynomial-size circuit instead of a uniform algorithm (that is, $\calP$ is equipped with non-uniform advice).
\end{definition}

\begin{definition}[Proof Complexity Generators~\cite{ABRW04,Krajicek04}]
    Let $s(n)<n$ be a function for seed length. A proof complexity generator is a map $C_n:\{0,1\}^s\to \{0,1\}^n$ computed by a family of polynomial-size circuits $\{C_n\}_n$. A generator is secure against a propositional proof system $P$ if for every large enough $n$ and every $y\in\{0,1\}^n$, $P$ does not have a polynomial-size proof of the (properly encoded) statement
    \begin{align*}
        \forall x\in \{0,1\}^s,C_n(x)\neq y.
    \end{align*}
\end{definition}

It is easy to see that the existence of proof complexity generators is closely related to the hardness of range avoidance. In fact we have:

\begin{theorem}[{Informal version of \cite[Theorem 6.6]{RenSW22}}]
    The range avoidance problem with suitable stretch is in $\FNP$ if and only if there exists a propositional proof system that breaks every proof complexity generator.
\end{theorem}

\paragraph{Pseudo-surjectivity.} In addition to the basic notion of hardness for proof complexity generators, several stronger notions have been proposed in the literature, including freeness~\cite{krajivcek2001tautologies}, iterability and pseudo-surjectivity~\cite{Krajicek04}, and $\bigvee$-hardness~\cite{krajivcek2024existence}. Pseudo-surjectivity is the strongest hardness notion among them. In this paper, we show that our proof complexity generators are pseudo-surjective in certain parameter regimes.

To motivate the definition of pseudo-surjectivity~\cite[Definition 3.1]{Krajicek04}, it is helpful to consider \emph{Student-Teacher games} for solving $\Avoid$. Let $G: \{0, 1\}^n \to \{0, 1\}^m$ be a mapping where $m > n$. A polynomial-time \emph{Student} attempts to find a string $y \in \{0, 1\}^m \setminus \Range(G)$ with the help of a \emph{Teacher} who has unbounded computational power. The game proceeds in rounds. In each round $i$, the Student proposes a candidate string $y_i \in \{0, 1\}^m$, and if $y_i \in \Range(G)$, the Teacher returns a preimage $q_i \in \{0, 1\}^n$ such that $G(q_i) = y_i$. If the Student ever proposes a string outside the range of $G$, they win the game.

A Student who attempts to solve $\Avoid$ in $k$ rounds can be represented as $k$ circuits $B_1, B_2, \dots, B_k$, where each $B_i$ uses the Teacher's responses from previous rounds to generate the next query. Specifically, $B_1$ outputs a fixed string $y_1 \in \{0, 1\}^m$, and each subsequent circuit $B_i$ ($i > 1$) takes the previous responses $q_1, q_2, \dots, q_{i-1} \in \{0, 1\}^n$ as inputs and outputs $y_i \in \{0, 1\}^m$. The game proceeds as follows:

\begin{itemize}
    \item The Student proposes $y_1 := B_1\in \{0, 1\}^m$. If $y_1 \not\in\Range(G)$, then the Student wins the game; otherwise, the Teacher returns some $q_1 \in \{0, 1\}^n$ such that $G(q_1) = y_1$.
    \item The Student then proposes $y_2 := B_2(q_1) \in \{0, 1\}^m$. If $y_2 \not \in \Range(G)$ then the Student wins the game; otherwise, the Teacher returns $q_2 \in \{0, 1\}^n$ such that $G(q_2) = y_2$.
    \item $\dots$
    \item This continues until round $k$, where the Student proposes $y_k := B_k(q_1, \dots, q_{k-1}) \in \{0, 1\}^m$. If $y_k \not\in \Range(G)$ then the Student wins the game; otherwise the Student loses the game.
\end{itemize}

To formally express whether the Student succeeds in the game, we define a formula stating that at least one of the Student's queries is outside the range of $G$. Let $B: \{0, 1\}^{n'} \to \{0, 1\}^m$ be a circuit, $z \in \{0, 1\}^{n'}$ and $x \in \{0, 1\}^n$ be disjoint variables, we define $\tau(G)_{B(z)}(x)$ to be the (properly encoded) statement that $B(z) \ne G(x)$. Then, using $\vec{q}_1, \dots, \vec{q}_{k-1} \in \{0, 1\}^n$ to represent the Teacher's responses, the Student wins if and only if
\begin{equation}
   \label{eq: the student wins the game} \bigvee_{i=1}^k\tau(G)_{B_i(q_1, \dots, q_{i-1})}(q_i).
\end{equation}

Roughly speaking, a generator $G$ is \emph{pseudo-surjective} for a proof system $\calP$ if $\calP$ cannot prove any Student wins the game, no matter how the Student is constructed. In other words, a generator $G$ is pseudo-surjective for $\calP$ if, for every sequence of Student circuits $(B_1, \dots, B_k)$, the formula~\eqref{eq: the student wins the game} is hard to prove in $\calP$.

Note that pseudo-surjectivity is indeed a stronger notion than ordinary hardness for proof complexity generators. Indeed, for every $y \in \{0, 1\}^m \setminus \Range(G)$, the trivial one-round Student with $B_1 = y$ clearly wins the game---yet pseudo-surjectivity implies that this fact is hard to prove in $\calP$.

We proceed to the formal definition. We also introduce the notion of \emph{$k$-round} pseudo-surjectivity, where the unprovability of \eqref{eq: the student wins the game} only holds for $k$-round Students for some fixed $k = k(n)$.

\begin{definition}[$k$-round pseudo-surjectivity~\cite{Krajicek04}]\label{def:pseudo-surjectivity}
    Let $\calP$ be any proof system, $G: \{0, 1\}^n \to \{0, 1\}^m$ be a circuit where $m > n$, and $s$ be a size parameter.\begin{itemize}
        \item We say that $G$ is \emph{$s$-pseudo-surjective} for $\calP$ if for every $k$ and every sequence of Student circuits $(B_1, B_2, \dots, B_k)$, \eqref{eq: the student wins the game} requires $\calP$-proof of size at least $s$.
        \item Fixing a parameter $k = k(n)$, we say that $G$ is \emph{$k$-round $s$-pseudo-surjective} for $\calP$ if for every sequence of Student circuits $(B_1, B_2, \dots, B_k)$, \eqref{eq: the student wins the game} requires $\calP$-proof of size $\ge s$.
    \end{itemize}
    
    When $s = n^{\omega(1)}$, we omit the parameter $s$ and simply say that $G$ is ($k$-round) pseudo-surjective for $\calP$.
\end{definition}

\subsection{Bounded Arithmetic}\label{sec: bounded arithmetic prelim}

Roughly speaking, $\PV_1$ is a theory of bounded arithmetic capturing ``polynomial-time'' reasoning. The language of $\PV_1$, $\calL(\PV)$, contains a function symbol for every polynomial-time algorithm $f: \N^k \to \N$, defined using Cobham's characterization of polynomial-time functions~\cite{cobham1964intrinsic}. Although Cook's $\PV$~\cite{Cook75} was originally defined as an equational theory (i.e., the only relation in $\PV$ is equality and there are no quantifiers), one can define a first-order theory $\PV_1$ by adding suitable induction schemes~\cite{Cook75, KRAJICEK1991143}. In the literature, the notation $\PV$ is often used to refer to the set of polynomial-time computable functions as well. The precise definition of $\PV_1$ is somewhat involved, and we refer the reader to the textbooks~\cite{Krajicek-book,Cook-Nguyen,Krajicek_proof_complexity} and references~\cite{Cook75,Jerabek06, ChenLiOliveira24, Li25}.

To capture reasoning in \emph{randomized} polynomial time, \Jerabek~\cite{Jerabek04,Jerabek-phd,Jerabek07} defined a theory $\APC_1$ by extending $\PV_1$ with the \emph{dual weak Pigeonhole Principle} for $\PV$ functions (i.e., polynomial-time functions). Let $\Eval(G,x):=G(x)$ be the circuit evaluation function. For a function $\ell(n) > n$, define $\dwPHP_\ell(\Eval)$ to be the following sentence

\begin{align*}
    &\dwPHP_\ell(\Eval):=\nonumber\nopagebreak\\
    &~~~~\forall n\in \mathsf{Log}~\forall\text{circuit }G:\{0,1\}^n\to \{0,1\}^{\ell(n)}~\exists y\in \{0,1\}^{\ell(n)}~\forall x\in \{0,1\}^n~[\Eval(G,x)\neq y]. 
\end{align*}

Here, ``$n\in\mathsf{Log}$'' is the standard notation in bounded arithmetic, which means that $n$ is the bit-length of some object; this notation allows us to reason about objects of size $\poly(n)$ instead of merely size $\polylog(n)$. The above sentence can be interpreted as the totality of $\Avoid$: every input $G: \{0, 1\}^n \to \{0, 1\}^{\ell(n)}$ has at least one solution $y$.

For this paper, it suffices to think of $\ell(n)$ as a large polynomial in $n$; in fact, under suitable hardness assumptions, we will be able to show that $\PV_1$ cannot prove $\dwPHP_\ell$ for every polynomial $\ell(n)$. This suffices to separate $\APC_1$ from $\PV_1$.

\paragraph{KPT witnessing and Student-Teacher games.} The only property of bounded arithmetic theories that we need in this paper is the \emph{KPT witnessing theorem}:

\begin{theorem}[{KPT Witnessing Theorem for $\PV_1$~\cite{KRAJICEK1991143}}]\label{thm:KPT}
    For every quantifier-free formula $\varphi(\vec{x},y,z)$ in the language $\calL(\PV)$, if $\PV_1 \vdash \forall \vec{x}~\exists y~\forall z~\varphi(\vec{x},y,z)$, then there is a $k\in \N$ and $\calL(\PV)$-terms $t_1,t_2,\dots,t_k$ such that
    \begin{equation}\label{eqn:KPT}
        \PV_1 \vdash \forall \vec{x}~\forall z_1~\forall z_2 \dots \forall z_k~\bigvee_{i=1}^k~\varphi(\vec{x},t_i(\vec{x},z_1,\dots,z_{i-1}),z_i).
    \end{equation}
\end{theorem}

In particular, \autoref{thm:KPT} implies that if $\PV_1\vdash \dwPHP_\ell(\Eval)$, then there exists a constant $k$ and a polynomial-time Student that wins the Student-Teacher game for the range avoidance problem. (Note that here the Student is computed by a \emph{uniform} algorithm that gets $(1^n, G)$ as inputs, as opposed to a family of non-uniform circuits in \autoref{sec: proof complexity}). To see this, let $\varphi((1^n, G), y, x) = 1$ iff $\Eval(G, x) \ne y$ and apply \autoref{thm:KPT}. We obtain a constant $k\in \N$ and $\calL(\PV)$-terms $t_1, t_2, \dots, t_k$ such that:
\begin{equation}\label{eq: KPT-Avoid}
    \PV_1\vdash \forall n\in\mathsf{Log}~\forall G~\forall z_1~\forall z_2 \dots \forall z_k~\bigvee_{i=1}^k \Eval(G, z_i) \ne t_i(G, z_1, \dots, z_{i-1}).
\end{equation}
This implies that the following Student wins the Student-Teacher game in $k$ rounds:\begin{enumerate}
    \item The Student and the Teacher are given a circuit $G: \{0, 1\}^n \to \{0, 1\}^{\ell(n)}$ as input.
    \item In the first round, the Student proposes $y_1 :=t_1(G)$ as a candidate non-output. If $y_1$ is correct (i.e., $\forall z_1~\varphi(G, y_1, z_1)$ is true), then the Student wins the game. Otherwise, the Teacher provides a counterexample $z_1$ such that $\Eval(G, z_1) = y_1$, i.e., a preimage $z_1\in G^{-1}(y_1)$.
    \item Then, the Student proposes a new candidate $y_2 := t_2(G, z_1)$ based on the counterexample given in the first round. If $y_2$ is a correct non-output, then the Student wins the game. Otherwise, the Teacher again provides a counterexample $z_2$ such that $\Eval(G, z_2) = y_2$, i.e., a preimage $z_2 \in G^{-1}(y_2)$.
    \item The game proceeds until the Student provides a correct witness $y$.
\end{enumerate}

\subsection{Extractors}
\begin{definition}[$k$-Source]
    A random variable $X$ is a \emph{$k$-source} if for every $x\in \mathrm{Supp}(X)$, $\Pr[X = x] \le 2^{-k}$.
\end{definition}

\begin{definition}[Strong Seeded Extractors]\label{def: comp-ext}	
    A polynomial-time computable function $\Ext: \{0, 1\}^n \times \{0, 1\}^d \to \{0, 1\}^m$ is a \emph{$(k, \eps)$-strong seeded extractor} if for every $k$-source $\calX$ over $\{0, 1\}^n$, the statistical distance of $(\calU_d, \Ext(X, \calU_d))$ and $(\calU_d, \calU_m)$ is at most $\eps$.
\end{definition}

Below is the only property of strong seeded extractors that we will use:

\begin{fact}
    Suppose $\Ext: \{0, 1\}^n \times \{0, 1\}^d \to \{0, 1\}^m$ is a $(k, \eps)$-strong seeded extractor. Then for every (possibly unbounded) adversary $\calA: \{0, 1\}^{d+m} \to \{0, 1\}$, the number of strings $x\in \{0, 1\}^n$ such that
	\begin{equation}\label{eq: extractor fact}
        \Pr_{r\sim \{0, 1\}^d}[\calA(r, \Ext(x, r)) = 1] < \Pr_{r\sim \{0, 1\}^d, z\sim \{0, 1\}^m}[\calA(r, z) = 1] - \eps
    \end{equation}
	is at most $2^k$.
\end{fact}
\begin{proof}[Proof Sketch]
	Fix an adversary $\calA$, let $\calX$ be the set of strings $x \in \{0, 1\}^n$ such that \eqref{eq: extractor fact} holds. (We abuse notation and also use $\calX$ to denote the uniform distribution over itself.) Note that $\calA$ distinguishes $\Ext(\calX, r)$ from the uniform distribution with advantage $\eps$. If $|\calX| \ge 2^k$, then the min-entropy of $\calX$ is at least $k$, contradicting the extractor properties of $\Ext$. Hence $|\calX| < 2^k$.
\end{proof}

This work requires extractors with exponentially small $\eps$, which can be constructed from any family of pairwise independent hash functions.

\begin{theorem}[Leftover Hash Lemma~\cite{ILL89}]\label{thm:lhl}
	Let $h: \{0, 1\}^n \times \{0, 1\}^d \to \{0, 1\}^m$ be a family of pairwise independent hash functions, where the first component (length $n$) is its input and the second component (length $d$) is its key. Then for every $k, \eps$ such that $m = k - 2\log(1/\eps)$, $h$ is a $(k, \eps)$-strong seeded extractor.

    In particular, if $n\geq m$ and $d\geq 2n$, there exists a family of pairwise independent hash functions $h$ that is $\F_2$-linear.\footnote{That is, for each fixed $r\in \{0, 1\}^d$, the function $h(-, r)$ is an $\F_2$-linear function over its inputs.} If we set $n\geq 3m+3$, $d\geq 2n$, $k:=n-1$, $\eps := 2^{-m-1}$, then $h$ is an $(n-1, \eps)$-strong seeded extractor.
\end{theorem}

\section{Hardness of Range Avoidance}\label{sec: Avoid notin FP}

\begin{theorem}\label{thm: demi-bit to Avoid notin FP}
	Assume that for some $m > n$, there exists a demi-bits generator $G: \{0, 1\}^n \to \{0, 1\}^N$ and $\Ext: \{0, 1\}^N \times \{0, 1\}^d \to \{0, 1\}^m$ is an $(N-1, 2^{-m-1})$-strong seeded extractor. ($N,d\leq\poly(m)$.) Then $\Avoid$ for polynomial-size circuits of stretch $n\to m$ is not in $\SearchNP$.  %
\end{theorem}
\begin{proof}
	Let $r \in \{0, 1\}^d$, define the circuit $C_r: \{0, 1\}^n \to \{0, 1\}^m$ with $r$ hardwired:
	\[C_r(s) = \Ext(G(s), r).\]
	Assume towards contradiction that there is a nondeterministic polynomial-time algorithm $\calA$ solving $\Avoid$. We construct the following nondeterministic adversary $\calB(y)$ that breaks the demi-bits generator $G$. Given an input $y \in \{0, 1\}^N$, the adversary $\calB$ accepts $y$ if and only if there exists $r \in \{0, 1\}^d$ such that some nondeterministic branch of $\calA(C_r)$ outputs $\Ext(y, r)$.

	It is easy to see that $\calB$ rejects every string $y\in\Range(G)$. To see this, suppose that $y = G(s)$ for some $s \in \{0, 1\}^n$. Then
	\[C_r(s) = \Ext(G(s), r) = \Ext(y, r),\]
	hence $\calA(C_r)$ will never output $\Ext(y, r)$.

	It remains to show that $\calB$ accepts at least $1/2$ fraction of strings $y \in \{0, 1\}^N$. For $r\in \{0, 1\}^d, z\in\{0, 1\}^m$, let $\calA'(r, z)$ be the adversary that outputs $1$ if there is a nondeterministic branch of $\calA(C_r)$ that outputs $z$, and outputs $0$ otherwise. Since $\Ext$ is a $(N-1, 2^{-m-1})$-strong seeded extractor, %
    the following is true for at least $1/2$ fraction of $y \in \{0, 1\}^N$:
	\[\Pr_{r\sim \{0, 1\}^d}[\calA'(r, \Ext(y, r)) = 1] \ge \Pr_{\substack{r\sim \{0, 1\}^d\\z\sim \{0, 1\}^m}}[\calA'(r, z) = 1] - 2^{-m-1} \ge 2^{-m-1}.\]
	For such $y$, there exists $r^\star\in \{0, 1\}^d$ and a nondeterministic branch of $\calA(C_{r^\star})$ that outputs $\Ext(y, r^\star)$.
	
	It follows that $\calB$ rejects every string in $\Range(G)$ but accepts at least $1/2$ fraction of strings $y\in \{0, 1\}^N$. This contradicts the security of $G$.
\end{proof}

\ThmDemiBitsPlusHash*
\begin{proof}
    We construct an extractor
    \(\Ext(y,h):=h(y)~(y\in\{0,1\}^N,h\in\calH)\). 
    According to \autoref{thm:lhl}, $\Ext\colon\{0,1\}^N\times\calH\to\{0,1\}^m$ is an $(N-1,2^{-m-1})$-strong seeded extractor. Therefore, by \autoref{thm: demi-bit to Avoid notin FP}, there exists $h\in\calH$ such that $\calA$ fails to solve $\Avoid$ on $\Ext(G(\cdot),h)$, i.e., $h\circ G$.
\end{proof}

\begin{corollary}
	Under \autoref{assumption: demi-bits in degree 2}, there are constants $\eps > 0$ and $d\ge 2$ such that $\Avoid$ for $\XOR\circ\AND_d$ circuits (i.e., degree-$d$ polynomials over $\F_2$) of stretch $n\mapsto n^{1+\eps}$ is not in $\SearchNP$.
\end{corollary}
\begin{proof}
	\autoref{assumption: demi-bits in degree 2} implies a demi-bits generator $G: \{0, 1\}^n \to \{0, 1\}^{n^{1+\delta}}$ and each output bit of $G$ can be computed by a degree-$d$ polynomial over $\F_2$, where $\delta > 0$ and $d\geq2$ are constants. Let $\eps := \delta/2$, $\Ext: \{0, 1\}^{n^{1+\delta}} \times \{0, 1\}^{2n^{1+\delta}} \to \{0, 1\}^{n^{1+\eps}}$ be a $(n^{1+\delta}-1,2^{-n^{1+\eps}-1})$-strong linear seeded extractor guaranteed by \autoref{thm:lhl}. Then, for every nondeterministic adversary $\calA$, there exists $r \in \{0, 1\}^{2n^{1+\delta}}$ such that $\calA$ fails to solve $\Avoid$ on the instance
	\[C_r(s) := \Ext(G(s), r).\]
	Since $\Ext$ is multi-linear and $G$ is a degree-$d$ polynomial over $\F_2$, $C_r$ is an $\XOR\circ\AND_d$ circuit.
\end{proof}

\subsection{From Demi-Bits Generators to Proof Complexity Generators}\label{subsec:pcg}
\def\hist{\mathsf{hist}}
Let $\calP$ be a proof system and $G: \{0, 1\}^n \to \{0, 1\}^\ell$ be a function computable in polynomial size where $\ell > n$. (We allow $G$ to take non-uniform advice.) Let $b \in \{0, 1\}^\ell$, denote as $\tau_b(G)$ the propositional formula encoding that $b$ is not in the range of $G$. We say $G$ is a:
\begin{itemize}
    \item \emph{demi-bits generator} against $\calP$, if for at least a $1/3$ fraction of $b \in \{0, 1\}^\ell$, $\calP$ does not have polynomial-size proof of $\tau_b(G)$; and $G$ is a
    \item \emph{proof complexity generator} against $\calP$, if for \emph{every} $b\in \{0, 1\}^\ell$, $\calP$ does not have polynomial-size proof of $\tau_b(G)$.
\end{itemize}

The precise definition of $\tau_b(G)$ as a $3$-CNF is as follows. The variables of $\tau_b(G)$ consist of $x\in \{0, 1\}^n$ and $\hist \in \{0, 1\}^s$, where $s$ is the number of internal gates in $G$ (including the output gates but not including the input gates). The intended meaning is that $G(x) = b$ and $\hist$ represents the values of internal gates of $G$ during the computation of $G(x)$. Each gate in $G$ corresponds to a bit $v_g$; if $g$ is an input (internal) gate then $v_g$ refers to some $x_i$ ($\hist_i$). For each internal gate $g \in G$ labeled by an operation $\circ_g$ (such as $\land$, $\lor$, or $\oplus$) and two children gates $g_l, g_r$, we have a constraint
\[v_g = v_{g_l} \mathbin{\circ_g} v_{g_r}\]
in $\tau_b(G)$. Similarly, for each $i\in [\ell]$ representing an output gate $g_i$, we have a constraint
\[v_{g_i} = b_i\]
in $\tau_b(G)$. Note that since every constraint only depends on at most $3$ variables, it can be written as a $3$-CNF of size at most $2^3 = 8$, and we can add every clause in this $3$-CNF into $\tau_b(G)$. We assume that the $\oplus$ gate of fan-in $2$ is included in our basis (looking ahead, it will be used to implement the extractor). The $3$-CNF $\tau_b(G)$ is simply the union of ($3$-CNFs generated from) these constraints over every internal and output gate $g \in G$.

Now we define the notion of \emph{simple parity reductions} between two CNFs. This is a technical notion that we need in \autoref{claim: parity reduction}.

\def\redu{\mathsf{redu}}

\begin{definition}\label{def: parity reduction}
    Let $F(x)$ and $G(y)$ be CNF formulas over variables $x = (x_1, \dots, x_n)$ and $y = (y_1, \dots, y_m)$. We say that there is a \emph{simple parity} reduction from $F$ to $G$, denoted as $F\le^\oplus G$, if:
    \begin{itemize}
        \item {\bf Variables.} The reduction is computed by a $\GF(2)$-linear mapping $\redu: \{0, 1\}^n \to \{0, 1\}^m$ (that is, every output bit of $\redu$ is the $\XOR$ of a subset of its input bits).
        \item {\bf Axioms.} For any clause $g\in G$, one of the following happens:\begin{compactitem}
            \item $g\circ\redu\equiv {\sf True}$;
            \item $g\circ\redu$ is equal to some axiom in $F$; or
            \item $g$ is a width-$1$ clause (i.e., one that consists of a \emph{single} literal) and $g\circ\redu$ is the $\XOR$ of a subset of axioms in $F$ (in which case these axioms in $F$ are also width-$1$ clauses).
        \end{compactitem}

    \end{itemize}
\end{definition}

We say a proof system $\calP$ is \emph{closed under simple parity reductions} if there is a polynomial $p$ such that the following holds. For every CNF $F$ and $G$, if there is a simple parity reduction from $F$ to $G$ and there is a length-$\ell$ $\calP$-proof of $G$, then there is a length-$p(\ell)$ $\calP$-proof of $F$.

We note that this notion is weaker than that of (degree-$1$) algebraic reductions in ~\cite{BussGIP01, DBLP:conf/stoc/RezendeGNPR021}. It follows from \cite[Lemma 8.3]{DBLP:conf/stoc/RezendeGNPR021} that many algebraic proof systems (such as Nullstellensatz and Polynomial Calculus) over $\GF(2)$ are closed under simple parity reductions when the complexity measure is degree. While we do not know if $\Res[\oplus]$ (resolution over linear equations modulo $2$~\cite{ItsyksonS20}) is closed under low-degree algebraic reductions, it is straightforward to prove that $\Res[\oplus]$ is closed under simple parity reductions (see~\autoref{sec: res parity}).

Recall that $G:\{0, 1\}^n \to \{0, 1\}^N$ is a purported demi-bits generator, $\Ext: \{0, 1\}^N\times \{0, 1\}^d \to \{0, 1\}^m$ is an extractor, and for a fixed $r\in\{0, 1\}^d$ we define the circuit $C_r: \{0, 1\}^n \to \{0, 1\}^m$ as
\[C_r(s) := \Ext(G(s), r).\]
We say that $\Ext$ is \emph{linear} if for every fixed randomness $r$, the function $\Ext(\cdot, r): \GF(2)^N \to \GF(2)^m$ is $\GF(2)$-linear. For every $r$ we fix a circuit $\Ext_r$ for computing $\Ext(\cdot, r)$ using $\oplus$ gates of fan-in $2$ only.

\begin{claim}\label{claim: parity reduction}
	Suppose that $\Ext$ is a linear extractor. For every $y\in\{0, 1\}^N$ and $r \in \{0, 1\}^d$, there is a simple parity reduction from $\tau_y(G)$ to $\tau_z(C_r)$, where $z := \Ext(y, r)$.
\end{claim}
First, as a sanity check, we show that $\tau_y(G)$ follows from $\tau_z(C_r)$ logically: Suppose that $\tau_y(G)$ is false and that $y = G(s)$ for some $s\in\{0, 1\}^n$, then
\[C_r(s) = \Ext(G(s), r) = \Ext(y, r) = z,\]
meaning that $\tau_z(C_r)$ is also false. Now we show that if $\Ext$ is a linear extractor, then the above deduction is actually a simple parity reduction under our formalization of $\tau_b(G)$:
\begin{proof}[Proof of \autoref{claim: parity reduction}]
	Recall that the variables of $\tau_y(G)$ consist of $s\in\{0, 1\}^n$ and $\hist_G\in\{0, 1\}^{|G|}$, where $|G|$ denotes the number of internal gates in $G$. Also, recall the variables of $\tau_z(C_r)$ consist of $s\in\{0, 1\}^n$ and $\hist_{C_r}\in\{0, 1\}^{|C_r|}$. Since $C_r(s) = \Ext(G(s), r)$, $\hist_{C_r}$ consists of $\hist_G$ as well as the internal gates of $\Ext(\cdot, r)$. Since $\Ext$ is linear, each internal gate in $\Ext(\cdot, r)$ is an XOR of variables in $\hist_G$. Therefore, one can compute a $\GF(2)$-linear map $\redu: \{0, 1\}^{n + |G|} \to \{0, 1\}^{n + |C_r|}$ that maps $(s, \hist_G)$ to $(s, \hist_{C_r})$.

	Now we show that for every clause $c \in \tau_z(C_r)$, one of the three cases in~\autoref{def: parity reduction} happens. Note that $c$ comes from an internal gate or an output gate of $C_r$.\begin{itemize}
		\item If $c$ comes from an internal gate of $G$, then $c\circ \redu$ (which is equal to $c$ itself) is an axiom in $\tau_y(G)$.
		\item If $c$ comes from an internal gate in $\Ext(\cdot, r)$, then $c\circ \redu \equiv {\sf True}$ by the definition of $\redu$.
		\item The only remaining case is that $c$ comes from an output gate. Suppose this is the $i$-th output gate of $C_r$ (where $i\in[m]$), and let $v'_i$ denote the variable (of $\tau_z(C_r)$) representing the $i$-th output of $C_r$. Note that $c$ is a width-$1$ axiom stating that $v'_i = z_i$.
        
        Let $S_i \subseteq [N]$ be such that $\Ext(y, r)_i = \bigoplus_{j\in S_i}y_j$. Then $\redu$ maps $v'_i$ to $\bigoplus_{j\in S_i}v^G_j$, where $v^G_j$ is the variable in $\tau_y(G)$ that represents the $j$-th output gate of $G$. We also have that $z_i = \bigoplus_{j\in S_i}y_j$. Hence $c\circ\redu$ is the $\XOR$ of the axioms $v^G_j = y_j$ over all $j\in S_i$. Since each $v^G_j = y_j$ is an axiom in $\tau_y(G)$, this concludes the proof. \qedhere
	\end{itemize}
\end{proof}

\begin{theorem}\label{thm:pcg}
    Let $\calP$ be a proof system closed under parity reductions. Let $G: \{0, 1\}^n \to \{0, 1\}^N$ be a demi-bits generator secure against $\calP$, and $\Ext: \{0, 1\}^N \times \{0, 1\}^d \to \{0, 1\}^m$ be an $(N-2, 2^{-m-1})$-strong linear seeded extractor. Then there is a non-uniform proof complexity generator secure against $\calP$.
\end{theorem}
\begin{proof}
    Suppose for contradiction that for every $r\in \{0, 1\}^d$, there exists a string $z(r) \in \{0, 1\}^m$ such that $\calP$ admits a length-$\ell$ proof of $\tau_{z(r)}(C_r)$, where $\ell \le \poly(|G|)$. For $r \in \{0, 1\}^d$ and $z\in \{0, 1\}^m$, let $\calA'(r, z)$ be the adversary that outputs $1$ if $\calP$ admits a length-$\ell$ proof of $\tau_z(C_r)$ and outputs $0$ otherwise. Since for every $r$, $\calA'(r, z(r)) = 1$, we have
    \[\Pr_{\substack{r\sim \{0, 1\}^d\\z\sim \{0, 1\}^m}}[\calA'(r, z) = 1] \ge 2^{-m}.\]
    Since $\Ext$ is a $(n-2, 2^{-m-1})$-strong extractor, for at least a $3/4$ fraction of $y \in \{0, 1\}^N$, we have
    \[\Pr_{r\sim \{0, 1\}^d}[\calA'(r, \Ext(y, r)) = 1] \ge \Pr_{\substack{r\sim \{0, 1\}^d\\z\sim \{0, 1\}^m}}[\calA'(r, z) = 1] - 2^{-m-1} > 0.\]
    Hence, for such $y \in \{0, 1\}^N$, there exists some $r := r(y)$ such that $\calP$ admits a length-$\ell$ proof of $\tau_z(C_r)$ where $z := \Ext(y, r)$. Since there is a parity reduction from $\tau_y(G)$ to $\tau_z(C_r)$ and $\calP$ is closed under parity reductions, it follows that $\calP$ admits a length-$\poly(\ell)$ proof of $\tau_y(G)$ as well, contradicting the security of $G$ as a demi-bits generator against $\calP$.
\end{proof}

Although super-polynomial lower bounds for $\Res[\oplus]$ remain open, it seems conceivable that we will eventually prove such lower bounds sooner or later. Our results suggest a potential approach for designing proof complexity generators against $\Res[\oplus]$: it suffices to design a \emph{demi-bits generator} against $\Res[\oplus]$ (which might be an easier task) and then apply \autoref{thm:pcg}.

\section{Lower Bounds for Student-Teacher Games}\label{sec:proof-complexity-app}

In this section, we show that $\Avoid$ is hard for Student-Teacher games. In \autoref{sec: separation bounded arithmetic} we prove lower bounds against uniform, polynomial-time Students, which implies a conditional separation between bounded arithmetic theories $\PV_1$ and $\APC_1$. In \autoref{sec: pseudosurjectivity}, we show that demi-bits generators can be transformed into proof complexity generators that are pseudo-surjective.

\subsection{Separating \texorpdfstring{$\PV_1$}{PV1} from \texorpdfstring{$\APC_1$}{APC1}}\label{sec: separation bounded arithmetic}

As discussed in \autoref{sec: bounded arithmetic prelim}, to separate $\PV_1$ from $\APC_1$, it suffices to show that there is no polynomial-time Student that wins the Student-Teacher game for $\Avoid$ in $O(1)$ rounds. In fact, we will show something stronger: Let $k = k(n)$ be a parameter, assuming the existence of demi-bits generators secure against $\AM/_{O(\log k)}$, there is no polynomial-time Student that wins the Student-Teacher game for $\Avoid$ in $k(n)$ rounds.

\begin{theorem}\label{thm:no-avoid^O(1)}
    Let $m, n, k$ be parameters such that $m > n$. Assume there exists a demi-bits generator $G: \{0,1\}^n \to \{0,1\}^N$ secure against $\AM/_{O(\log k)}$. Let $\Ext: \{0,1\}^N \times \{0,1\}^d \to \{0,1\}^m$ be an $(N-1, 2^{-10km})$-strong extractor. Then for every deterministic polynomial-time Student $A$, there is a string $r\in \{0,1\}^d$ and a Teacher such that $A$ fails to solve $\Avoid$ on $C_r$ in $k$ rounds, where $C_r: \{0,1\}^n \to \{0,1\}^m$ is the circuit
    \[C_r(s) := \Ext(G(s), r).\]
\end{theorem}
\begin{proof}
    Let $A$ denote the Student algorithm where $A(i, C, z_1, \dots, z_{i-1})$ outputs the $i$-th candidate solution. For strings $s_1, \dots, s_j\in \{0,1\}^n$ (where $j\le k$), we say that $(s_1, \dots, s_j)$ is a \emph{valid trace} for $A$ on the input $C_r$ if all of the following are true:
    \begin{compactitem}
        \item $C_r(s_1) = A(1, C_r)$ (that is, $s_1$ is a valid counterexample for $A(1, -)$);
        \item $C_r(s_2) = A(2, C_r, s_1)$ (that is, $s_2$ is a valid counterexample for $A(2, -)$);
        \item $\dots$
        \item and $C_r(s_j) = A(j, C_r, s_1, s_2, \dots, s_{j-1})$ (that is, $s_j$ is a valid counterexample for $A(j, -)$).
    \end{compactitem}
    
    We prove the following stronger claim that implies \autoref{thm:no-avoid^O(1)}:
    
    \begin{claim}\label{claim:induction}
        For every $j\le k$, there exist $s_1, s_2, \dots, s_j \in \{0,1\}^n$ such that
        \[\Pr_{r\sim \{0,1\}^d}[(s_1, \dots, s_j)\text{ is a valid trace for $A$ on the input $C_r$}] \ge 2^{-2jm}.\]
    \end{claim}
    Clearly, \autoref{claim:induction} implies \autoref{thm:no-avoid^O(1)} by setting $j := k$ and noticing that $2^{-2jm} > 0$.

    We prove~\autoref{claim:induction} by induction on $j$. The base case $j=0$ is trivially true. Now we assume the claim is true for $j-1$, which gives strings $s_1,\dots,s_{j-1}$ such that
    \[\Pr_{r\sim \{0,1\}^d}[(s_1, \dots, s_{j-1})\text{ is a valid trace for $A$ on the input $C_r$}] \ge 2^{-2(j-1)m}.\]
    Consider the following $\AM/_{O(\log k)}$ protocol that attempts to break the demi-bits generator $G$. This protocol has the index $j$ hardwired as advice but is otherwise uniform.

\begin{algorithm2e}[H]
      \caption{The $\AM/_{O(\log k)}$ protocol $\mathcal{P}$ breaking demi-bits generator $G$} 
      \label{alg:am}
      \SetKwInput{KwSetting}{Setting}
	\SetKwInput{KwPara}{Parameters}
	\SetKwInput{KwAssumption}{Assumption}
	\SetKwInput{KwAdvice}{Advice}
	\SetKw{KwAccept}{accept}
	\SetKw{KwReject}{reject}
	
	\KwIn{A string $y\in \{0,1\}^N$.}
    	$\Prover$ sends $s_1,\dots,s_{j-1}$\label{line: prover sends s}\; 
        $\Prover$ and $\Verifier$ run the Goldwasser--Sipser protocol to estimate\label{line: prover and verifier runs Goldwasser Sipser}
        \begin{algomathdisplay}
            p := \Pr_{r\sim \{0,1\}^d}\left[\begin{aligned}
                &(s_1, \dots, s_{j-1})\text{ is a valid trace for $A$ on the input $C_r$}\\
                \text{and }&\Ext(y, r) = A(j, C_r, s_1, \dots, s_{j-1}).
            \end{aligned}\right]
        \end{algomathdisplay}
    	If $p \ge 2^{-(2j-1)m-1}$ then $\Verifier$ accepts; if $p \le 2^{-(2j-1)m-2}$ then $\Verifier$ rejects\;
\end{algorithm2e}

\paragraph{Completeness of $\mathcal{P}$.} We show that for at least a $1/2$ fraction of $y$, there is a $\Prover$ such that the $\Verifier$ accepts w.p.~$\ge 2/3$ in $\calP$. In the first round, the honest $\Prover$ sends $(s_1, \dots, s_{j-1})$ as guaranteed by the induction hypothesis. Recall that this means
\[\Pr_{r\sim \{0,1\}^d}[(s_1, \dots, s_{j-1})\text{ is a valid trace for $A$ on the input $C_r$}] \ge 2^{-2(j-1)m}.\]

Let $\mathsf{Test}(r, z) = 1$ if $(s_1, \dots, s_{j-1})$ is a valid trace for $A$ on the input $C_r$ and $z = A(j, C_r, s_1, \dots, s_{j-1})$, and $\mathsf{Test}(r, z) = 0$ otherwise. Clearly, we have
\[\Pr_{r\sim \{0,1\}^d, z\sim \{0,1\}^m}[\mathsf{Test}(r, z) = 1] \ge 2^{-2(j-1)m} / 2^m = 2^{-(2j-1)m}.\]
Since $\Ext$ is an $(N-1,2^{-10km})$-strong extractor, for $\ge 1/2$ fraction of $y$'s, it holds that 
\[\Pr_{r\sim \{0,1\}^d}[\mathsf{Test}(r, \Ext(y, r)) = 1] \ge \Pr_{r\sim \{0,1\}^d,z\sim \{0,1\}^m}[\mathsf{Test}(r,z) = 1] - 2^{-10km}\ge 2^{-(2j-1)m-1}.\]

It follows that there is a $\Prover$ for the Goldwasser--Sipser protocol in \autoref{line: prover and verifier runs Goldwasser Sipser} of~\autoref{alg:am} such that the verifier accepts with probability at least $2/3$.

\paragraph{Employing the lack of soundness.} Since $G$ is a demi-bits generator secure against $\AM/_{O(\log k)}$ adversaries, $\mathcal{P}$ does not have the soundness for all sufficiently large $n$. In other words, there is a $\Prover^*$ that makes the $\Verifier$ accepts some $y\in\Range(G)$ with probability $>1/3$. Fix such a string $y$, let $s_j$ be any $n$-bit string such that $G(s_j) = y$, and let $s_1, \dots, s_{j-1} \in \{0,1\}^n$ be the message sent in \autoref{line: prover sends s} (of~\autoref{alg:am}) by $\Prover^*$ on the input $y$.

Since $\Verifier$ accepts with probability $>1/3$, by the soundness of the Goldwasser--Sipser Protocol (\autoref{lemma:GS-protocol}), we have that
\[\Pr_{r\sim \{0,1\}^d}\left[\begin{aligned}
                &(s_1, \dots, s_{j-1})\text{ is a valid trace for $A$ on the input $C_r$}\\
                \text{and }&\Ext(y, r) = A(j, C_r, s_1, \dots, s_{j-1}).
            \end{aligned}\right] \ge 2^{-j(2m+1)-2}.\]
Recall that $\Ext(y, r) = C_r(s_j)$, hence the above condition inside $\Pr_{r\sim \{0,1\}^d}[\cdot]$ means exactly that $(s_1, \dots, s_j)$ is a valid trace for $A(C_r)$. This implies \autoref{claim:induction} for $j$.
\end{proof}

We remark that the parameters we obtained in \autoref{thm:no-avoid^O(1)} are (almost) tight in the following sense. \autoref{thm:no-avoid^O(1)} showed that (under plausible assumptions) for every parameter $k\le \poly(n)$, there is no deterministic polynomial-time Student that wins the Student-Teacher game for $\Avoid$ in $k$ rounds, when given an $\Avoid$ instance of size $s = \poly(k, n) > k$. On the other hand, under plausible derandomization assumptions, for every size parameter $s$, there exists a deterministic polynomial-time Student that wins the game on size-$s$ circuits within $k = \poly(s, n) > s$ rounds~\cite[Appendix A]{DBLP:conf/stoc/IlangoLW23}. 

Finally, setting $k=O(1)$ in~\autoref{thm:no-avoid^O(1)}, we obtain the following separation:

\ThmPVvsAPC*

\subsection{From Demi-Bits to Pseudo-Surjectivity}\label{sec: pseudosurjectivity}

\begin{theorem}\label{thm:pseudo-pcg-fix}
    Let $G:\{0,1\}^n\to\{0,1\}^N$ be a demi-bits generator secure against $\NP/_{\poly}$, $k\in\N$, and $\Ext:\{0,1\}^N\times\{0,1\}^d\to\{0,1\}^m$ be an $(N-1,\eps)$-strong linear extractor for $\eps:=2^{-10km}$. ($k,d,N\leq\poly(m)$.) Then for every non-uniform propositional proof system $\calP$, there is a string $r\in\{0,1\}^d$ such that the circuit $C_r: \{0,1\}^n\to \{0,1\}^m$ defined as
    \[C_r(s) := \Ext(G(s), r)\]
    is a non-uniform $k$-round pseudo-surjective proof complexity generator secure against $\calP$.
\end{theorem}

\begin{proof}
    Suppose, for contradiction, that such an $r\in \{0,1\}^d$ does not exist. Then, for any $r\in\{0,1\}^d$, there exist Student circuits
    \[
        B^{(r)}=\left\{B_i^{(r)}\colon \{0,1\}^{n(i-1)}\to\{0,1\}^m\right\}_{i\in[k+1]}
    \]
    such that $\calP$ admits a proof of
    \[
        \bigvee_{i=1}^k\tau(C_r)_{B_{i}(q_1,q_2,\dots,q_{i-1})}(q_i).
    \]
    (i.e., $\calP$ can prove that $B^{(r)}$ wins the Student-Teacher game on $C_r$.)

    Now we attempt to break the demi-bits generator $G$. For $j=0,1,\dots,k$, define
    \[
        \Phi_j:=\max_{(s_1,s_2,\dots,s_j)\in\{0,1\}^{nj}}\Pr_{r\sim\{0,1\}^m}\left[\begin{aligned}
            &\text{$\exists$ Student $B$ such that:}\\
            &\text{(1) $\calP$ can prove that $B$ wins the Student-Teacher game on $C_r$;}\\
            &\text{(2) $B_i(s_1,\dots,s_{i-1})=C_r(s_i)$ for all $i\in[j]$.}
        \end{aligned}\right].
    \]
    (Item (2) says that the history of the Student-Teacher game in the first $i$ rounds is exactly $s_1,\dots,s_j$.) We make the following claims on the values of $\Phi_0$ and $\Phi_k$:
    \begin{itemize}
        \item $\Phi_0=1$: When $j=0$, item (2) obviously holds, and item (1) holds by our assumption that $C_r$ is not a pseudo-surjective proof complexity generator;
        \item $\Phi_k=0$: When $j=k$, for any $r,B$ items (1) and (2) cannot hold simultaneously. This is because (2) implies that $B$ loses the Student-Teacher game, which contradicts (1).
    \end{itemize}\par
    Simple calculations show that there exists $j\in\{0,1,\dots,k-1\}$ such that 
    \[
        \Phi_j\cdot 2^{-m}-\eps>2\cdot\Phi_{j+1}.
    \]
    We use such $j$ to break the demi-bits generator $G$. Let $(s_1^*,\dots,s_j^*)$ be the tuple such that the maximum is achieved in the definition of $\Phi_j$, i.e.,
    \[
        \Phi_j=\Pr_{r\sim\{0,1\}^m}\left[\begin{aligned}
            &\text{$\exists$ Student $B$ such that:}\\
            &\text{(1) $\calP$ can prove that $B$ wins the Student-Teacher game on $C_r$;}\\
            &\text{(2) $B_i(s_1^*,\dots,s_{i-1}^*)=C_r(s_{i}^*)$ for all $i\in[j]$.}
        \end{aligned}\right].
    \]
    
    Consider the following algorithm: on an input $y\in\{0,1\}^n$, let
    \[
        p(y):=\Pr_{r\sim\{0,1\}^m}\left[\begin{aligned}
            &\text{$\exists$ Student $B$ such that:}\\
            &\text{(1) $\calP$ can prove that $B$ wins the Student-Teacher game on $C_r$;}\\
            &\text{(2) $B_i(s_1^*,\dots,s_{i-1}^*)=C_r(s_i^*) = \Ext(G(s_{i}^*),r)$ for all $i\in[j]$,}\\
            &\text{~~~~~and $B_{j+1}(s_1^*,\dots,s_j^*)=\Ext(y,r)$.}
        \end{aligned}\right].
    \]
    Our algorithm accepts $y$ if $p(y)\geq\Phi_j\cdot2^{-m}-\eps$, and rejects if $p(y)\leq\Phi_{j+1}$. This can be implemented by the Goldwasser--Sipser set lower bound protocol since $\Phi_j\cdot2^{-m}-\eps>2\cdot\Phi_{j+1}$, and the condition inside $\Pr_{r\sim\{0,1\}^m}[\cdot]$ is certifiable in polynomial time with the help of a prover. Finally, we prove that this algorithm breaks demi-bits generator $G$:
    \begin{itemize}
        \item {\bf For any $y\in \Range(G)$, we have $p(y)\leq\Phi_{j+1}$:}\par
        Suppose $y=G(s)$. Then
        \begin{align*}
            p(y)&=\Pr_{r\sim\{0,1\}^m}\left[\begin{aligned}
                &\text{$\exists$ student $B$ such that:}\\
                &\text{(1) $\calP$ can prove that $B$ wins the Student-Teacher game on $\Ext(G(\cdot),r)$;}\\
                &\text{(2) $B_i(s_1^*,\dots,s_{i-1}^*)=\Ext(G(s_{i}^*),r)$ for all $i\in[j]$,}\\
                &\text{~~~~~and $B_{j+1}(s_1^*,\dots,s_j^*)=\Ext(G(s),r)$.}
            \end{aligned}\right]\\
            &\leq\Phi_{j+1}.
        \end{align*}
        Where the $\leq$ in the second line follows from the definition of $\Phi_{j+1}$.
        \item {\bf For half of $y\in \{0,1\}^n$, we have $p(y)\geq\Phi_j\cdot2^{-m}-\eps$:}\par
        For simplicity, we use ``$\mathsf{Test}(r,\Ext(y,r))$'' to denote the condition inside $\Pr_{r\sim\{0,1\}^m}[\cdot]$ in the definition of $p(y)$. We have:
        \[
            \Pr_{\substack{r\sim\{0,1\}^m\\z\sim\{0,1\}^m}}\left[\mathsf{Test}(r,z)\right]=\Phi_j\cdot 2^{-m}.
        \]
        By the definition of strong extractors, for half of $y\in\{0,1\}^n$, we have
        \[
            \Pr_{\substack{r\sim\{0,1\}^m}}\left[\mathsf{Test}(r,\Ext(y,r))\right]\geq\Pr_{\substack{r\sim\{0,1\}^m\\z\sim\{0,1\}^m}}\left[\mathsf{Test}(r,z)\right]-\eps=\Phi_j\cdot 2^{-m}-\eps.\qedhere
        \]
    \end{itemize}
\end{proof}

\begin{corollary}\label{cor: pseudo-surjective}
    Suppose for any parameter $N\leq \poly(n)$, there exists a demi-bits generator $G: \{0,1\}^n \to \{0,1\}^N$ secure against $\NP/_\poly$. Then for every non-uniform propositional proof system $\calP$ and any parameters $k\leq\poly(n)$ and $n<m<\poly(n)$, there is a circuit $C: \{0,1\}^n \to \{0,1\}^m$ of size $\poly(n)$ such that $C$ is a $k$-round pseudo-surjective proof complexity generator secure against $\calP$. %
\end{corollary}
\begin{proof}
    In \autoref{thm:pseudo-pcg-fix}, let $N:=100km$ and $\Ext\colon\{0,1\}^N\times\{0,1\}^{O(km)}\to\{0,1\}^{m}$ be an $(N-1,2^{-10km})$-strong seeded extractor guaranteed by \autoref{thm:lhl}. Then there exists $r\in\{0,1\}^{O(km)}$ such that $C_r:=\Ext(G(\cdot),r)$ is a $k$-round pseudo-surjective proof complexity generator secure against $\calP$. %
\end{proof}

\section{Hardness of Range Avoidance from Predictable Arguments}\label{sec:ILW-interpretation}

Ilango, Li, and Williams~\cite{DBLP:conf/stoc/IlangoLW23} proved $\Avoid\not\in\FP$ assuming the existence of JLS-secure $\iO$ and $\NP\ne\coNP$. Given our main result that the existence of demi-bits generators implies the hardness of $\Avoid$, it is natural to ask whether the $\iO$ assumption used in~\cite{DBLP:conf/stoc/IlangoLW23} can be weakened to a ``Minicrypt'' assumption. In particular, we conjecture:
\begin{conjecture}\label{conj: hardness of Avoid from OWF}
    $\Avoid\not\in\FP$ follows from the existence of \emph{one-way functions} and $\NP\ne\coNP$.
\end{conjecture}

Unfortunately, we are not able to prove \autoref{conj: hardness of Avoid from OWF}. Instead, in this section, we present an alternative interpretation of~\cite{DBLP:conf/stoc/IlangoLW23}. Although the results and proofs are not new, we hope that our new perspective helps make progress towards proving~\autoref{conj: hardness of Avoid from OWF}, or in general, basing $\Avoid\not\in\FP$ from minimal assumptions.

\paragraph{Witness encryptions.} Instead of $\iO$, what~\cite{DBLP:conf/stoc/IlangoLW23} actually needs is \emph{witness encryption}~\cite{GargGSW13}. Let $L\in \NP$ (usually we take $L$ to be an $\NP$-complete language such as $\SAT$). A \emph{witness encryption scheme} for $L$ is a pair of algorithms $(\Enc, \Dec)$ with the following interface:\footnote{Witness encryption schemes as defined in~\cite{GargGSW13} can only encrypt a one-bit message $m\in \{0, 1\}$. To obtain a witness encryption scheme that can encrypt arbitrarily long message, one can simply encrypt each message bit independently. The security of this new scheme follows easily from a hybrid argument.}
\begin{itemize}
	\item Given an instance $x \in \{0, 1\}^n$, a message $m$, a security parameter $1^\lambda$, and some random coins $r$, $\Enc(x, b, 1^\lambda; r)$ outputs the encryption of $m$ under $x$.
	\item Given an instance $x \in \{0, 1\}^n$, a witness $w$ that $x\in L$, a ciphertext $ct$, and the security parameter $1^\lambda$, $\Dec(x, w, ct, 1^\lambda)$ outputs the encrypted message $m$. (We assume that $\Dec$ is deterministic.)
\end{itemize}

We require $(\Enc, \Dec)$ to be \emph{correct} and \emph{$2^{-\lambda^\eps}$-secure} for some constant $\eps > 0$; here we say that
\begin{itemize}
    \item $(\Enc, \Dec)$ is \emph{(perfectly) correct} if for every $x\in L$, witness $w$ for $x$, bit $b$, and randomness $r$, it is the case that 
    \[\Dec(x, w, \Enc(x, b, 1^\lambda; r), 1^\lambda) = b.\]
    \item $(\Enc, \Dec)$ is \emph{$\delta(\cdot)$-secure} if for every $x\not\in L$, message $m$, and every non-uniform adversary $\calA$ of size $\poly(n)$,
    \[\mleft|\Pr[\calA(\Enc(x, m, 1^\lambda)) = 1] - \Pr[\calA(\calU) = 1]\mright| < \delta(n).\]
\end{itemize}

\paragraph{Hardness of $\Avoid$ from predictable arguments.} Witness encryption implies the following \emph{predictable argument system}~\cite{DBLP:conf/pkc/FaonioN017} for $L$. Let $x\in \{0, 1\}^n$ be an instance known to both the prover and the verifier. Recall that in an argument system, the prover is computationally bounded. If $x\in L$, then the prover also has access to a witness $w \in \{0, 1\}^m$. Let $\ell > m$ be a parameter and $\lambda := \ell^{2/\eps}$.
\begin{itemize}
	\item First, the $\Verifier$ picks a random message $m\sim \{0, 1\}^\ell$, encrypts $m$ as $ct \sim \Enc(x, m, 1^\lambda)$, and sends $ct$ to the $\Prover$.
	\item Then the $\Prover$ sends a message $m'$. In particular, the honest $\Prover$ sends $m' \gets \Dec(x, w, ct, 1^\lambda)$.
	\item The $\Verifier$ accepts if and only if $m' = m$.
\end{itemize}

Now, suppose that there is a deterministic algorithm $\calA$ solving $\Avoid$, we show that $L \in \coNP$. The idea is to \underline{use $\calA$ as the prover in the above argument system}. In particular, let $C_{x, ct}: \{0, 1\}^m \to \{0, 1\}^\ell$ denote the circuit that given $w\in \{0, 1\}^m$ as input, outputs $\Dec(x, w, ct, 1^\lambda)$. Upon receiving $ct$, $\Prover$ always sends $\calA(C_{x, ct})$ to $\Verifier$.

\begin{claim}
	If $x\in L$, then the $\Prover$ will never convince the $\Verifier$.
\end{claim}
\begin{proof}
	The only message that convinces the verifier is $\tilde{m} := \Dec(x, w, ct, 1^\lambda)$, where $w$ is a witness of $x\in L$. Clearly, $\tilde{m}$ is in the range of $C_{x, ct}$.
\end{proof}

\begin{claim}
	If $x\not\in L$, then the $\Prover$ has a non-zero probability of convincing the $\Verifier$.
\end{claim}
\begin{proof}
	Let $m^\star \in \{0, 1\}^\ell$ be any string such that $\calA(ct) = m^\star$ with probability at least $2^{-\ell}$ over a truly random $ct$. Since $\calA$ runs in polynomial time (which means the security of witness encryption holds for $\calA$), the probability over $ct \sim \Enc(x, m^\star, 1^\lambda)$ that $\calA(ct) = m^\star$ is at least $2^{-\ell} - 2^{-\lambda^\eps} > 0$. 
\end{proof}

The above claims imply a nondeterministic algorithm for deciding the complement of $L$. On input $x \in \{0, 1\}^n$, guess the $\Verifier$'s first message $m$ and randomness $r$, compute $ct := \Enc(x, m, 1^\lambda; r)$, and accept if the $\Verifier$ accepts when the $\Prover$ replies with $\calA(C_{x, ct})$.

In summary, witness encryption implies that $\NP$ has a \emph{special type of predictable argument system}. Moreover, if $\Avoid \in \FP$, then plugging the $\Avoid$ algorithm as the prover results in the following intriguing situation: if $x\in L$, then the $\Verifier$ \emph{never} accepts, while if $x\not\in L$, then the $\Verifier$ accepts with \emph{non-zero} probability! This allows us to put every language with such argument systems in $\coNP$.

An interesting question is to identify the weakest possible argument system for $\NP$ such that the above situation happens. Can we build such argument systems using only one-way functions? Such an argument system would make progress towards proving~\autoref{conj: hardness of Avoid from OWF}.

\section*{Acknowledgments}
\ifnum\Anonymity=0
We thank Yilei Chen for helpful discussions regarding the LPN assumption, Xin Li for helpful discussions about extractors, and Rahul Ilango for sending us a draft version of~\cite{Ilango25}.
\else
Anonymous acknowledgments.
\fi

{\small \bibliography{ref}}

\appendix

\section{Ilango's Proof}\label{sec: RAHUL}
We present an alternative proof of our main results that demi-bits generators imply $\Avoid\not\in\SearchNP$. This proof is due to Rahul Ilango~\cite{Ilango25} in an independent and concurrent work.

\begin{theorem}[{Essentially~\cite{Ilango25}}]\label{thm: Ilango25}
    Let $G: \{0, 1\}^n \to \{0, 1\}^m$ be a candidate demi-bits generator and $t := 30m$. Suppose that $m > n + \lceil \log t\rceil$. For strings $s_1, s_2, \dots, s_t\in \{0, 1\}^m$, define the generator
    \[C_{\vec{s}}(x, i) := G(x) \oplus s_i.\]
    If there is a nondeterministic algorithm $\calA$ that solves $\Avoid$ on instances of the form $C_{\vec{s}}$, then there is a nondeterministic adversary $\calB$ that breaks the demi-bits generator $G$.
\end{theorem}

The proof relies on the following fact underlying Lautemann's proof~\cite{Lautemann83} that $\BPP$ is in the polynomial hierarchy.
\begin{fact}\label{fact: Lautemann}
    Let $R\subseteq \{0, 1\}^m$ be a set of size at least $\eps \cdot 2^m$, and $t := 10m\eps^{-1}$. Then, with probability at least $1-2^{-m}$ over the strings $s_1, s_2, \dots, s_t \sim \{0, 1\}^m$, it holds that for every string $z\in\{0, 1\}^m$, there exist $y\in R$ and $i\in[t]$ such that $z = y \oplus s_i$.
\end{fact}
\begin{proof}[Proof Sketch]
    Fix $z\in \{0, 1\}^m$, then the probability over $s_1, s_2, \dots, s_t\sim\{0, 1\}^m$ that $z\ne y\oplus s_i$ for every $y\in R$ and $i\in [t]$ is at most $(1-\eps)^t < 4^{-m}$. \autoref{fact: Lautemann} follows from a union bound.
\end{proof}

\def\Rej{\textsc{Rej}}
\begin{proof}[Proof of \autoref{thm: Ilango25}]
    Let $\calA$ be a nondeterministic adversary such that $\calA(s_1, \dots, s_t)$ outputs a string not in the range of $C_{\vec{s}}$. Let $\calB$ be the nondeterministic adversary $\calB$ that on input $y \in \{0, 1\}^m$, accepts if there exists $s_1, \dots, s_t$ and $i\in [t]$ such that $\calA(s_1, \dots, s_t) = s_i \oplus y$, and rejects otherwise. 
    
    We argue that $\calB$ breaks the demi-bits generator $G$. Clearly, for every $y\in \Range(G)$, every $\vec{s}$, and every $i\in[t]$, letting $y = G(x)$, then $s_i\oplus y = C_{\vec{s}}(x, i) \in \Range(C_{\vec{s}})$, hence $\calB(y)$ always rejects.

    On the other hand, let $\Rej := \{y \in \{0, 1\}^m: \calB(y)\text{ rejects}\}$, it remains to prove that $|\Rej|$ is small. Suppose for contradiction that $|\Rej| \ge 2^m / 3$, then by~\autoref{fact: Lautemann}, there exist strings $s_1, s_2, \dots, s_t \in \{0, 1\}^m$ such that for every string $z \in \{0, 1\}^m$, there exists a string $y \in \Rej$ such that $z = s_i\oplus y$ for some $i\in [t]$. Needless to say, this holds for $z = \calA(s_1, \dots, s_t)$, hence $\calB(y)$ should accept when its nondeterministic guesses are equal to $(s_1, \dots, s_t)$. This is a contradiction to our assumption that $y \in \Rej$. Hence we have $|\Rej| \le 2^m/3$ and $\calB$ indeed accepts most of its inputs.
\end{proof}

\section{Candidate Demi-Bits Generators}\label{sec: candidate demi-bits}

\subsection{Demi-Bits Generators with Polynomial Stretch}
\autoref{assumption: demi-bits with poly stretch} (demi-bits generators with polynomial stretch) follows from Rudich's conjecture on the unprovability of circuit lower bounds~\cite{Rudich97}, with the \emph{truth table generator} $\TT$ being a candidate demi-bits generator.

\def\lb{\mathsf{lb}}
For a Boolean function $f: \{0, 1\}^n \to \{0, 1\}$ and a parameter $s(n) \le \poly(n)$, let $\lb(f, s)$ denote the sentence stating that ``$f$ requires circuit complexity at least $s(n)$''. This sentence can be written as a CNF of size $2^{O(n)}$ and, if true, admits a trivial proof of length $2^{\tilde{O}(s(n))}$ in most reasonable proof systems. \emph{Rudich's conjecture} asserts that there is no (non-uniform) propositional proof system that has length-$2^{O(n)}$ proofs of $\lb(f, s)$ for a large fraction of Boolean functions $f$. This can be equivalently formulated as the non-existence of ``$\NP/_{\poly}$-natural proofs'' against polynomial-size circuits (see~\cite{Rudich97} for more details). Rudich's conjecture was further investigated in~\cite{PichS19, SanthanamT21}.

It is easy to see that Rudich's conjecture is equivalent to the demi-hardness of the ``truth table generator'' $\TT:\{0, 1\}^{O(s\log s)} \to \{0, 1\}^{2^n}$: Given as input the description of a size-$s$ circuit $C$, $\TT(C)$ outputs the truth table of $C$. As we only want demi-bits generators with polynomial stretch, we can set the parameter $s(n)$ to be $2^{\eps n}$ for some constant $\eps > 0$. %

Rudich's conjecture follows from the existence of \emph{super-bits} generator $g: \{0, 1\}^\ell \to \{0, 1\}^{\ell+1}$~\cite{Rudich97,DBLP:journals/jacm/GoldreichGM86}.\footnote{This is true for $s(n) = 2^{\Omega(n)}$. If super-bits generators with \emph{subexponential} security exist, then $\TT$ is a secure demi-bits generator even for $s(n) = \poly(n)$.} It is open whether Rudich's conjecture also follows from the existence of any demi-bits generator $g: \{0, 1\}^\ell \to \{0, 1\}^{\ell + 1}$, i.e., whether $\TT$ is the ``most secure'' demi-bits generator (see Open Problem 4 of~\cite{Rudich97}). Hence, we have:

\begin{proposition}
    \autoref{assumption: demi-bits with poly stretch} follows from either:\begin{itemize}
        \item Rudich's conjecture on the unprovability of random circuit lower bounds; or
        \item the existence of super-bits generators.
    \end{itemize}
\end{proposition}

\subsection{Constant-Degree Demi-Bits Generators from Learning Parity with Noise}

\paragraph{Learning Parity with Noise} is the assumption that noisy linear equations over $\F_2$ are hard to solve. Let $n\in \N$ be the number of variables, $m := m(n)\in \N$ be the number of equations, and $\mu := \mu(n) \in (0, 1)$ be the noise rate. Here we work in the regime where $m = n^{1+\eps}$ and $\mu = n^{-\eps}$ for some constant $\eps > 0$. Let $A \sim \F_2^{m\times n}$ be a random matrix, $\vec{s}\in \F_2^n$ be a hidden random vector (i.e., the solution), and $\vec{e}\sim {\sf Ber}(\mu)^m$ be a hidden noise vector where each entry is equal to $1$ w.p.~$\mu$ independently. The LPN assumption asserts that the following two distributions are computationally indistinguishable:
\[(A, A\vec{s} + \vec{e})\quad\text{v.s.}\quad (A, \calU_m).\]
Roughly speaking, we will assume the above indistinguishability holds even for nondeterministic adversaries, in the sense that no non-uniform proof system can efficiently prove a vector $\vec{v} \in \F^m$ is \emph{not} of the form $A\vec{s} + \vec{e}$ when $\vec{e}$ is a $(\mu\cdot m)$-sparse vector. That is:

\begin{assumption}\label{assumption: demi-hardness of LPN}
    For some (public) matrix $A\in \F_2^{m\times n}$, there is no polynomial-size non-uniform nondeterministic circuit $C: \{0, 1\}^m \to \{0, 1\}$ such that $C$ accepts a constant fraction of random strings but rejects every string of the form $A\vec{s} + \vec{e}$, where $\vec{e} \in \F_2^m$ is $(\mu\cdot m)$-sparse and $\vec{s} \in \F_2^n$.
\end{assumption}

\begin{fact}
    \autoref{assumption: demi-hardness of LPN} implies \autoref{assumption: demi-bits in degree 2}.
\end{fact}
\begin{proof}
    Let $d := \lceil 2/\eps\rceil$. Since $\mu\cdot m < \frac{m^{1-1/d}}{d}$, by~\cite[Lemma 3.1]{GajulapalliGNS23}, there exists a polynomial-time computable function $f: \F_2^{O(\mu m^{1+1/d})} \to \F_2^m$ whose range contains all vectors of sparsity at most $s$, such that each output of $f$ is a degree-$d$ polynomial.
    
    Now consider the following generator $g: \F_2^{O(\mu m^{1+1/d}) + n}\to \F_2^m$. The input of $g$ consists of $\overrightarrow{e_{\sf enc}} \in \F_2^{O(\mu m^{1+1/d})}$ and $\vec{s} \in \F_2^n$. The output is $A\vec{s} + f(\overrightarrow{e_{\sf enc}})$ (where $A$ is hardwired in the circuit computing $g$). It is easy to see that every output bit of $g$ is computable by a degree-$d$ polynomial over $\F_2$, and the range of $g$ contains every vector of the form $A\vec{s} + \vec{e}$ where $\vec{e}$ is $(\mu\cdot m)$-sparse and $\vec{s} \in \F_2^n$. The input length of $g$ is $O(\mu m^{1+1/d}) + n \le O(n^{1+\eps/2})$, which is polynomially smaller than the output length $m = n^{1+\eps}$.
\end{proof}

\subsection{Constant-Degree Demi-Bits Generators from Goldreich's Generator}

\paragraph{Goldreich's generator} is an influential candidate pseudorandom generator with large stretch that is computable with constant locality (i.e., in $\NC^0$)~\cite{Goldreich11-PRG}. To define the generator $G: \{0, 1\}^n \to \{0, 1\}^m$, fix a $d$-uniform hypergraph with $n$ vertices and $m$ (ordered) hyperedges each of size $d$, and a predicate $P: \{0, 1\}^d \to \{0, 1\}$. For each $i\in[m]$, the $i$-th output bit is obtained by applying $P$ to the input bits on the $i$-th hyperedge. That is, suppose the $i$-th hyperedge contains vertices $v_{i, 1}, v_{i, 2}, \dots, v_{i, d}$, then on input $x \in \{0, 1\}^n$, the $i$-th output bit is
\[G(x)_i := P(x_{v_{i, 1}}, x_{v_{i, 2}}, \dots, x_{v_{i, d}}).\]

It seems plausible to conjecture that Goldreich's generator is a secure demi-bits generator when instantiated with a suitable hypergraph and a suitable predicate $P$:
\begin{assumption}\label{assumption: Goldreich}
    There exist constants $c, d > 1$, a predicate $P: \{0, 1\}^d \to \{0, 1\}$, and a non-uniform family of $d$-hypergraphs $\{\calG_n\}_{n\in\N}$ with $n^c$ hyperedges such that Goldreich's PRG $G: \{0, 1\}^n \to \{0, 1\}^{n^c}$ instantiated with $P$ and $\{\calG_n\}$ is a secure demi-bits generator.
\end{assumption}

Clearly, \autoref{assumption: Goldreich} implies \autoref{assumption: demi-bits in degree 2}.

In fact, the following predicate $P: \F_2^5\to \F_2$ is frequently considered in the literature:
\[P_{\rm MST06}(x_{1\sim 5}) := x_1 + x_2 + x_3 + x_4x_5.\]
Note that $P_{\rm MST06}$ is a degree-$2$ function over $\F_2$. It was shown in~\cite{DBLP:journals/rsa/MosselST06} that some instantiation of Goldreich's PRG with $P_{\rm MST06}$ fools linear tests. Instantiations using this predicate was also conjectured to be secure against polynomial-time adversaries in~\cite{DBLP:conf/tcc/BogdanovKR23}. Instantiating Goldreich's PRG with the predicate $P_{\rm MST06}$ and a suitable family of $5$-hypergraphs gives us a candidate demi-bits generator computable by \emph{degree-$2$ polynomials over $\F_2$}.

To summarize, we have:
\begin{proposition}[Informal]
    \autoref{assumption: demi-bits in degree 2} follows from either:\begin{itemize}
        \item the demi-hardness of Learning Parity with Noise in certain parameter regime; or
        \item the demi-hardness of certain suitably instantiated Goldreich's generator.
    \end{itemize}
\end{proposition}

\section{\texorpdfstring{$\Res[\oplus]$}{Res[2]} Is Closed Under Simple Parity Reductions}\label{sec: res parity}

\begin{theorem}
    $\Res[\oplus]$ is closed under simple parity reductions. That is, let $F(x_{1\sim n}) = f_1\land f_2\land \dots \land f_m$ and $G(y_{1\sim {n'}}) = g_1\land g_2 \land \dots \land g_{m'}$ be CNF formulas such that $F\le^\oplus G$. If there exists $\Res[\oplus]$ refutation of $G$ in $s$ steps, then there exists a $\Res[\oplus]$ refutation of $F$ in $2nm' + s$ steps.
\end{theorem}
\begin{proof}
    We assume familiarity with $\Res[\oplus]$ (the definition can be found in~\cite{ItsyksonS20}).
    
    Let $C_1, C_2, \dots, C_s$ be a $\Res[\oplus]$ refutation of $G$ where each $C_i$ is a disjunction of linear equations modulo $2$, $C_i = g_i$ for every $1\le i\le m$, and $C_s = \bot$. Let $\redu: \{0, 1\}^n \to \{0, 1\}^m$ be the simple parity reduction from $F$ to $G$. Define $C'_i = C_i\circ \redu$, then $C'_i$ is still a disjunction of linear equations modulo $2$. It is easy to see that $C'_1, C'_2, \dots, C'_s$ is still a valid $\Res[\oplus]$ derivation, and that $C'_s = \bot$. Hence there is an $s$-step $\Res[\oplus]$ refutation from the axioms $\{g_i\circ\redu\}_{1\le i\le m}$.

    It suffices to show that each $g_i\circ\redu$ can be proved from the axioms of $F$. This is easy to see when $g_i\circ\redu\equiv {\sf True}$ or $g_i\circ\redu$ is equal to some $f_i$, hence we only need to consider the third case in~\autoref{def: parity reduction} where $g_i\circ\redu$ is the $\XOR$ of some axioms in $f_i$. Note that one can derive $(a\oplus b = 0)$ from $(a = 0)$ and $(b = 0)$ in $2$ steps\footnote{First weaken $(b = 0)$ to derive $(a = 1\lor a\oplus b = 0)$, then resolve $(a = 0)$ and $(a = 1\lor a \oplus b = 0)$ to derive $(a\oplus b = 0)$.}, hence $g_i\circ\redu$ can be derived from $F$ in $2n$ steps. Since there are at most $m'$ linear clauses of the form $g_i\circ\redu$ that need to be derived, the total number of steps is at most $2nm' + s$.
\end{proof}

\end{document}